\documentclass[12pt]{article}
\usepackage{amsmath,epsfig,amssymb,amsfonts,amsthm,epsfig,verbatim}
\usepackage{color,graphicx,subfig,lscape,longtable,xr}
\usepackage{multirow,soul}
\usepackage[authoryear, sort]{natbib}
\usepackage{tikz, float}

\newcommand{\indep}{\;\, \rule[0em]{.03em}{.67em} \hspace{-.25em}\rule[0em]{.65em}{.03em} \hspace{-.25em}\rule[0em]{.03em}{.67em}\;\,}

\newcommand{\sumi}{\ensuremath{\sum_{i=1}^{n}}}
\newcommand{\sumj}{\ensuremath{\sum_{j=1}^{n}}}

\newtheorem{Th}{\underline{\bf Theorem}}

\newtheorem{Rem}{\underline{\bf Remark}}

\newtheorem{Lem}{\underline{\bf Lemma}}
\newtheorem{Cor}{\underline{\bf Corollary}}

\def\boxit#1{\vbox{\hrule\hbox{\vrule\kern6pt  \vbox{\kern6pt#1\kern6pt}\kern6pt\vrule}\hrule}}
\def\bse{\begin{eqnarray*}}
\def\ese{\end{eqnarray*}}
\def\be{\begin{eqnarray}}
\def\ee{\end{eqnarray}}
\def\bsq{\begin{equation*}}
\def\esq{\end{equation*}}
\def\bq{\begin{equation}}
\def\eq{\end{equation}}

\def\var{\hbox{var}}
\def\cov{\hbox{cov}}

\def\wh{\widehat}
\def\wt{\widetilde}

\def\mR{\mathbb{R}}
\def\n{\nonumber}

\def\cov{\mbox{cov}}

\def\sumi{\sum_{i=1}^n}
\def\sumj{\sum_{j=1}^n}
\def\trans{^{\rm T}}

\def\calT{\mbox{ $\mathcal{T}$}}
\def\ba{{\boldsymbol\alpha}}

\def\bb{{\boldsymbol\beta}}
\def\bd{{\boldsymbol\delta}}

\def\bg{{\boldsymbol\gamma}}
\def\bphi{{\boldsymbol\phi}}
\def\bzeta{{\boldsymbol\zeta}}

\def\A{{\bf A}}
\def\a{{\bf a}}
\def\B{{\bf B}}
\def\c{{\bf c}}
\def\C{{\bf C}}
\def\D{{\bf D}}
\def\V{{\bf V}}
\def\g{{\bf g}}
\def\f{{\bf f}}

\def\b{{\bf b}}
\def\I{{\bf I}}

\def\bS{{\bf S}}

\def\U{{\bf U}}

\def\X{{\bf X}}

\def\x{{\bf x}}

\def\bSig{{\bf \Sigma}}

\def\log{\hbox{log}}

\def\squarebox#1{\hbox to #1{\hfill\vbox to #1{\vfill}}}
\def\btheta{{\boldsymbol \theta}}

\def\0{{\bf 0}}
\def\1{{\bf 1}}

\def\var{\hbox{var}}
\def\cov{\hbox{cov}}

\def\pr{\hbox{pr}}
\def\wh{\widehat}
\def\wt{\widetilde}

\def\log{\hbox{log}}

\def\logit{{\mbox{logit}}}

\begin{document}
\thispagestyle{empty}
\baselineskip=28pt
\vskip 5mm
\begin{center}
{\LARGE{\bf Covariate balancing for causal inference on categorical and continuous treatments}}
\vskip 1mm 
{\bf 
Seong-ho Lee$^1$, Yanyuan Ma$^1$ and Xavier de Luna$^2$}

$^1$Department of Statistics, Pennsylvania State University, USA \\
and
$^2$Department of Statistics/USBE, Ume\r{a} University, Sweden
\vspace{.55cm} 
\end{center}
\baselineskip=22pt

\begin{quotation}
	\noindent {\it Abstract:} We propose novel
        estimators for categorical and continuous treatments by
        using an optimal 
covariate balancing strategy for inverse probability weighting. 
The resulting estimators are shown to be consistent
  and asymptotically normal for causal contrasts of interest, either when the model explaining
 treatment assignment is correctly
 specified, or when the correct set of bases for the outcome models
 has been chosen and the assignment model is
 sufficiently rich. 
For the categorical treatment case, we show that the estimator
 attains the semiparametric efficiency bound when all models are correctly specified.  For
 the continuous case, the causal parameter of interest is a
 function of the treatment dose. The latter is not parametrized and 
the estimators proposed are shown to have bias 
and variance of the classical nonparametric rate. 
Asymptotic
 results are complemented with simulations illustrating the finite
 sample properties. Our analysis of a data set suggests a nonlinear
 effect of BMI  on the decline in self reported health.
\end{quotation}

\section{Introduction}

Encouraged by the recent booming development of the causal inference literature, 
we devise and study a novel inference tool for categorical and continuous
treatments by using covariate balancing strategies for inverse probability weighting
\citep[e.g.,][]{ImaiRatkovic:14, wangzubi:20,
  Fanetal:2018}. Our study is built on the
   fundemental idea on optimal covariate balancing of
  \cite{Fanetal:2018}, while we overcome additional methodological and
  theoretical challenges.
  
  When estimating a causal effect
  on an outcome, weighting based on the propensity score (model for the probability
  of the treatment given observed pre-treatment covariates)
  is often used to construct optimal estimators by an augmentation using fitted models for the outcome given the covariates. These augmented inverse probability weighting estimators have robustness
  properties to the specification of models used, and are locally
  efficient \citep[e.g.,][]{robins_rot1995,Scharfsteinetal:1999}. 
A vast majority of the literature on causal inference have focused on binary
treatments, i.e. where the causal parameter of interest is a
contrast between two treatments.
Nevertheless, there is an increasing
interest in multi-valued treatments \cite[e.g.,][]{fong2018,kennedyetal:2017,yangetal:16} as often encountered in applied
work, both in the medical and social sciences.
Causal
effects of categorical treatment were formalized by, e.g.,
\cite{Imbens:2000} and \cite{Robins:2000}, while \cite{CATTANEO2010}
deduced the semiparametric efficiency bound; see also
\cite{yangetal:16} for a review.  
Causal effects of continuous treatments were formalized in, e.g.,
\cite{Robins:2000}, \cite{vdlaanrobins:03}, \cite{hiranoimbens:04} and
\cite{galvaowang:15}. In contrast to previous works,
\cite{kennedyetal:2017} proposed a double robust estimation strategy
avoiding parametric specification of the dose-response curve. 
 
We contribute to the somewhat less rich literature on  robust estimation
for categorical and continuous treatments by using an
estimation strategy based on 
covariate balancing propensity score estimation for inverse probability weighting
\citep[e.g.,][]{ImaiRatkovic:14,fong2018}. 
\cite{Fanetal:2018} recently obtained key results in the binary treatment case by
  specifying which covariate functions should be balanced for
  efficient inference: the 
  propensity score model should be fitted through balancing a set of bases for the
  outcome models in the space spanned by the covariates.  We provide
corresponding results to the categorical and
  continuous treatment cases, hence completes the story. In particular, the procedures we
proposed balance the ``most suitable''   
functions of the covariates when the propensity score is correctly specified, in the sense that 
they minimize the variability of the causal effect estimation.
 When
the propensity score is misspecified and the outcome basis functions are
correct, the procedure looks for an approximate balance by minimizing the
squared bias of the resulting estimator.
As other recent proposals for the binary treatment
case \citep[][]{wangzubi:20,atheyetal:18,zubi:15,wongchan:17}, the method
presented here does not necessarily try to achieve exact balance when this
is not possible, although in practice exact balance can always be
targeted by enriching the assignment model.

For both the categorical and continuous treatment case, the
 proposed estimators are shown to be robust, i.e. consistent and asymptotically normal for
 causal contrasts of interest, either when the model explaining
 treatment assignment is correctly
 specified, or when the correct set of bases for the outcome models
 has been chosen and the propensity score model is
 sufficiently rich.
 For the
 categorical treatment case, we show that the estimator proposed
 attains the semiparametric efficiency bound when both the treatment
 assignment model and the outcome basis are correctly specified. For
 the continuous case, the causal parameter of interest is a
 function. The latter is not parametrized 
 and the estimators proposed are shown to have bias 
and variance of the classical nonparametric order
under typical regularity
 conditions, hence with a usual bias-variance trade-off.

The rest of the paper is organized as follows. Sections \ref{sec:cat}
and \ref{sec:cont} deal with the categorical and the continuous treatment
cases, respectively. In both sections, inverse probability weighting
estimators are introduced, where a working model for the generalized
propensity score is estimated by balancing basis functions for the
outcome models. We establish the theoretical
properties of the estimators. Simulation studies are conducted in
Section \ref{sec:simu} to illustrate the finite sample performance of
our methods. In Section \ref{sec:app}, we estimate the
dose-response curve of BMI  
  on the decline in self reported health from baseline to a 9 year
  follow up in a population of ages 50 or older. Section 
\ref{sec:disc} concludes the paper, while all proofs are relegated to the
Appendix.

\section{Categorical treatments}\label{sec:cat}
\subsection{Balancing scores and preliminaries on estimation}
Consider $K+1$ treatments, $A=0, 1,  \dots, K$, and their respective
potential outcomes $Y^0, \dots, Y^K$. 
We observe a random sample $(A_i, Y_i, \X_i), i=1, \dots, n$, where
we assume
$Y_i=Y_i^k$ if $A_i=k$, and $\X_i\in \mR^d$ is a vector of
pre-treatment covariates. We also assume ignorability of the treatment
assignment, i.e. 
$E(Y_i^k\mid\X_i,A_i)=E(Y_i^k\mid\X_i)\equiv m(k,\X_i)$ and
$\pr(A_i=k\mid\X_i=\x)\equiv\pi_0(k,\x)>\delta>0$ for all $k\in\{0,
1,\dots, K\}$ and all $\x$, where $\pi_0(k,\x)$ is named
generalized propensity score in the literature
\citep{Imbens:2000}. 

Let $\theta_k\equiv E(Y_i^k)$ for $k=0, 1,\dots, K$ be the average
response to the different treatment levels. The parameters of
interest are typically average causal effects between treatment
levels, i.e. causal contrasts such as $\theta_k-\theta_0$, if $k=0$ is a
treatment level of 
reference. We consider a parametric
working model $\pi(k,\x,\bb)$ for $\pi_0(k,\x)$, with $\bb\in \mR^p$,
and vectors of basis 
functions, $\B(k,\X):\mR^{d+1}\to \mR^q$, aiming at spanning
$m(k,\x)$. We assume $q$ does not depend on $k$ for notational simplicity.
Thus, correct specification will imply that there exists a value $\bb_0$ with
\be\label{eq:pispec}
\pi(k,\x,\bb_0)=\pi_0(k,\x),
\ee
 and there exists $\ba=(\ba_0\trans, \dots,\ba_K\trans)\trans$ with
\be\label{eq:mspec}
\ba_{k}\trans\B(k,\x)=m(k,\x),
\ee
  for all $k$ and all $\x$. Misspecification, i.e. situations when
(\ref{eq:pispec}) or (\ref{eq:mspec})   
does not hold for any value of $\bb$ and
  $\ba$, will also be considered in the sequel. Note that one of the
  advantages of the herein studied balancing approach is that the
  parameter $\ba$ does not need to be known or estimated. We hence do not use a subscript $0$ on
  $\ba$ and $m(\cdot)$ to distinguish true parameter value and correct
  model since this will be clear from the context.

For estimating $\theta_k$ under the above assumptions one needs to
control for the covariates $\X_i$ by using one or both 
working models. In particular, $\pi(k,\x,\bb)$ is a balancing score in
the sense that $\X_i \indep A_i \mid \pi(k,\x,\bb_0)$ under
(\ref{eq:pispec}) \citep{RR1983}.  
 Thus, for the binary case ($K=1$), \cite{ImaiRatkovic:14} proposed to solve
$$\sumi\left\{\frac{I(A_i=1)}{\pi(1,\X_i,\bb)}-\frac{I(A_i=0)}{\pi(0,\X_i,\bb)}\right\}\b(\X_i)=0,$$
where $\b(\X_i)$ is a vector valued function of the covariates. 
Based on the
resulting fitted propensity score $\pi(k,\X_i,\wh\bb)$, an inverse
probability weighting estimator for $\theta_k$ is
\be\label{eq:thetacat}
\wh\theta_k=n^{-1}\sumi
\frac{I(A_i=k)Y_i}{\pi(k,\X_i,\wh\bb)}.
\ee 
Two issues arise regarding the above procedure. One is that if
the propensity score model (\ref{eq:pispec}) is misspecified 
then, $\wh\theta_k$ is generally biased. Two is the choice of $\b(\X)$, which is
largely left unsupervised. 
\cite{Fanetal:2018} overcome these two issues in the binary case ($K=1$),
and proposed an optimal choice for  $\b(\X)$, 
in the sense that the resulting treatment effect estimator 
is consistent when 
(\ref{eq:pispec}) is correct, or when  (\ref{eq:mspec}) is correct
and (\ref{eq:pispec}) has sufficient flexibility, 
and is efficient if both are correct.

We aim to achieve the same kind of optimality and robustness 
in the categorical treatment case. Two different estimators may
be introduced with different 
properties, which we discuss heuristically below, before giving a
formal treatment in the next section. The first possibility to estimate
$\bb$ is to solve the following balancing condition 
\be\label{eq:betacat0}
\sumi\left[ 
\left\{\frac{I(A_i=k)}{\pi(k,\X_i,\bb)}-1\right\}\B(k,\X_i)
-
\left\{\frac{I(A_i=0)}{\pi(0,\X_i,\bb)}-1\right\}\B(0,\X_i)
\right]=\0
\ee 
at all $k=1, \dots, K$, i.e. a system of $qK$ equations.
GMM, as described below, can be used if $qK\geq p$. 
This balancing 
condition is motivated by pushing the bias of the contrast estimator
$\wh\theta_k-\wh\theta_0$ towards zero.   
In fact, it will be shown that the asymptotic bias  of
$\wh\theta_k-\wh\theta_0$ is equal to
$$
E\left[\left\{{I(A_i=k)}/{\pi(k,\X_i,\bb)}-1\right\}m(k,\X_i)-\left\{{I(A_i=0)}/{\pi(0,\X_i,\bb)}-1\right\}m(0,\X_i)\right].
$$

An alternative to setting the bias of $\wh\theta_k-\wh\theta_0$ to
zero for $k=1,\ldots,K$, is to directly put the bias of $\wh\theta_k$ to zero, for
$k=0,\ldots,K$, by separately balancing both terms in
(\ref{eq:betacat0}), i.e. solving the condition 
\be\label{eq:betacat}
\sumi
\left\{\frac{I(A_i=k)}{\pi(k,\X_i,\bb)}-1\right\}\B(k,\X_i)=\0
\ee   
at all $k=0, \dots, K$, i.e. a system of $q(K+1)$ equations. We will use GMM allowing for $q(K+1)\geq p$; see (\ref{eq:gmm}) below. 

The two choices are not necessarily equivalent. In fact, the former choice allows 
for biased estimation of $\wh\theta_k$ with the only aim to estimate the contrast
$\theta_k-\theta_0$ without bias. 
We find that, if $\wh\theta_k$ is indeed biased, then
$\wh\theta_k-\wh\theta_0$ will not be efficient. 
This is because local efficiency holds when the
the fitted propensity score is correctly specified and its parameters
are consistently estimated, which is not the case when
(\ref{eq:betacat}) does not hold. Due to this consideration,
below we focus on solving (\ref{eq:betacat}) and show that the
resulting estimator of $\theta_k$ has, under certain conditions, a
robust property and, when all working models are correctly specified,
reaches the asymptotic semiparametric efficiency 
bound.

\subsection{Asymptotic properties}

We now establish a robustness property and the asymptotic
distribution results of the estimator in
(\ref{eq:thetacat}), where $\bb$ is estimated through covariate
balancing (\ref{eq:betacat}); see Appendix \ref{app:cat} for proofs.
To gain an intuitive understanding of
the robustness property, we can verify that when the propensity score
model is correctly specified, i.e. when
(\ref{eq:pispec}) holds for all $k$ and all $\x$, $\wh\bb$ is
$\sqrt{n}-$consistent under the standard regularity conditions for GMM
estimation \citep{newey:mcfadden:94}, and 
$\pi(k,\x,\wh\bb)\to\pi(k,\x,\bb_0)=\pi_0(k,\x)$ 
in probability as $n$ tends to infinity. The consistency is a
consequence of 
$$E\left[
\left\{{I(A_i=k)}/{\pi(k,\X_i,\bb_0)}-1\right\}\B(k,\X_i)\right]=\0$$
in combination with the regularity conditions, irrespective of whether a correct basis for
the outcome models is specified.
This then leads to the convergence of
\bse
E(\wh\theta_k)=E\left\{
n^{-1}\sumi 
\frac{I(A_i=k)Y_i}{\pi(k,\X_i,\wh\bb)}\right\}
\to
E\left\{
\frac{I(A_i=k)Y_i^k}{\pi_0(k,\X_i)}\right\}
=\theta_k,
\ese
as $n\rightarrow\infty $. 
On the other hand, when the outcome model basis is actually correctly specified, i.e.
when (\ref{eq:mspec}) holds for all $k$ and $\x$, then the propensity
model (\ref{eq:pispec}) does not need be correct 
as long as (\ref{eq:betacat}) has a solution. In such case,
$\wh\bb$ is consistent for some value $\bb^*$, hence $\pi(k,\x,\wh\bb)$ converges to
some function $\pi(k,\x)$ in probability.
We then have
\bse
E(\wh\theta_k)&=&E\left\{
n^{-1}\sumi 
\frac{I(A_i=k)Y_i}{\pi(k,\X_i,\wh\bb)}\right\}
\to E\left\{
\frac{I(A_i=k)Y_i^k}{\pi(k,\X_i)}\right\}\n\\
&=&
E\left[\left\{
\frac{\pi_0(k,\X_i) }{\pi(k,\X_i)}-1\right\}m(k,\X_i)+m(k,\X_i)\right]
=\theta_k,
\ese
as $n\rightarrow\infty $, where the last equality is the result of (\ref{eq:mspec}) and
(\ref{eq:betacat}).

To be more formal, let 
\bse
\f_{ki}(\bb)\equiv\left\{\frac{I(A_i=k)}{\pi(k,\X_i,\bb)}-1\right\}\B(k,\X_i)
,
\ese
 $\f_i(\bb)\equiv\{\f_{0i}(\bb)\trans,
 \dots,\f_{Ki}(\bb)\trans\}\trans$,
 $\V(\bb)\equiv E\{ \f_i(\bb)\f_i(\bb)\trans\}$,
 $\wh\V(\bb)\equiv n^{-1}\sumi \f_i(\bb)\f_i(\bb)\trans$,
$\A(\bb)\equiv E\left\{ \partial\f_i(\bb)/\partial\bb\trans\right\}$
and
$\wh\A(\bb)\equiv n^{-1}\sumi\partial\f_i(\bb)/\partial\bb\trans$.
Further, let $\btheta\equiv(\theta_0, \dots, \theta_K)\trans$, 
$g_{ki}(\bb)\equiv {I(A_i=k)Y_i}/{\pi(k,\X_i,\bb)}
-E\{m(k,\X_i)\}$, $\g_i(\bb)=\{g_{1i}(\bb),\dots,
g_{Ki}(\bb)\}\trans$ and
$\B(\bb)\equiv E\{\partial \g_i(\bb^*)/\partial{\bb^*}\trans\}$.
We solve for 
a solution of (\ref{eq:betacat}) by minimizing
\be\label{eq:gmm}
\{\sumi\f_i(\bb)\}^T\wh\V(\bb)^{-1}\{\sumi\f_i(\bb)\}.
\ee

We will use the following regularity conditions.
\begin{description}
\item[A0.] $\bb^*$ is the unique solution of $E\{\f_i(\bb)\}=0$. 
\item[A1.]  The variance-covariance matrix $\V(\bb^*)$ has bounded
  positive eigenvalues. 
\item[A2.] $\f_i(\bb)$ is differentiable with respect to $\bb$.
\item[A3.] The matrix $\A(\bb^*)$  is bounded and has full column rank.
\item[A4.]	 $\g_i(\bb)$ is differentiable with respect to $\bb$.
\end{description}

These are classical regularity conditions. Condition A0 requires the
existence and uniqueness of a solution, where the uniqueness can be
relaxed to local uniqueness. The existence requirement is automatic
when the $\pi(k,\x,\bb)$ model is correct. In this case
$\bb^*=\bb_0$. It is also natural and standard  when $(K+1)q$, the 
number of equations in $E\{\f_i(\bb)\}$ is not larger than $p$, the
dimension of $\bb$, which is achievable through enriching 
the $\pi(k,\x,\bb)$ model. Thus, regardless of whether $\pi(k,\x,\bb)$ is 
correctly specified or not, we can always justify Condition A0.

\begin{Th}\label{thm:asympt}
Assume that either
(\ref{eq:pispec}) holds for all $k$ and $\x$, or (\ref{eq:mspec}) holds for all $k$ and $\x$.
Then, under regularity conditions A0 to A4,
$n^{1/2}(\wh\btheta-\btheta)$ has asymptotic normal distribution with
mean zero and variance 
\bse
\bSig&=&
\B(\bb^*)
\{\A(\bb^*)\trans\V(\bb^*)^{-1}\A(\bb^*)\}^{-1}\B(\bb^*)\trans
+\C(\bb^*)\\
&&-\B(\bb^*)\{\A(\bb^*)\trans\V(\bb^*)^{-1}\A(\bb^*)\}^{-1}
\A(\bb^*)\trans\V(\bb^*)^{-1}\D(\bb^*)\\
&&-\D(\bb^*)\trans
[\B(\bb^*)\{\A(\bb^*)\trans\V(\bb^*)^{-1}\A(\bb^*)\}^{-1}
\A(\bb^*)\trans\V(\bb^*)^{-1}]\trans,
\ese
where
$\C(\bb^*)\equiv E\{\g_i(\bb^*)^{\otimes2}\}$ and
$\D(\bb^*)\equiv E\{\f_i(\bb^*)\g_i(\bb^*)\trans\}$.
\end{Th}

Theorem \ref{thm:asympt} highlights
a robust property. On the one hand, 
if the propensity score is correctly specified then
we will have a consistent estimator of the treatment contrast even if
the outcome basis is misspecified.
On the other hand, we can also afford to misspecify the propensity
score model, provided that the outcome basis functions are correctly
specified. In the latter case, Condition A0 plays a pivotal role and
it is crucial to ensure it. 
An example is to use the model $\pi(k,\x,\bb)=\bb_{(k)}\trans\B(k,\x)$,
$k=0, \dots, K$, with $\bb=(\bb_{(0)}\trans,\cdots,\bb_{(K)}\trans)\trans$
so that $\bb$ has length $p=q(K+1)$.  Then 
(\ref{eq:betacat}) is the derivative of the loss function  
 \be\label{eq:samebasis}
 \sum_{i=1}^n [I(A_i=k)\log \{\bb_{(k)}\trans\B(k,\X_i)\} -
 \bb_{(k)}\trans\B(k,\X_i)],
\ee 
for $k=0, \dots, K$, hence the minimizer is a root of 
(\ref{eq:betacat}).
 The utilization of the same basis of functions for both 
 nuisance models is used in \cite{wangzubi:20} as well. To further
 accommodate one's favorite propensity model, we can also make linear
 combination of this model and any candidate model in mind.

The asymptotic variance simplifies greatly when all models are
correctly specified, and a local efficiency result is obtained. 

\begin{Cor}\label{cor:asympt}
Assume that (\ref{eq:pispec}) and (\ref{eq:mspec}) hold for all $k$ and $\x$ and
let $\var(Y_i^k\mid\X_i)=v(k,\X_i)$. Then, under the regularity
conditions of Theorem \ref{thm:asympt},  
 $n^{1/2}(\wh\btheta-\btheta)$ has asymptotic normal distribution with
 mean zero and variance 
 $$\bSig=\C(\bb_0)-\B(\bb_0)
\{\A(\bb_0)\trans\V(\bb_0)^{-1}\A(\bb_0)\}^{-1}\B(\bb_0)\trans,$$
where
\bse
\A_k(\bb_0)&=&E\left\{
-\frac{\B(k,\X_i)\pi'_\bb(k,\X_i,\bb_0)\trans
}{\pi(k,\X_i,\bb_0)}\right\}, \\
\B_{k}(\bb_0)&=&E\left\{
-\frac{m(k,\X_i)\pi'_\bb(k,\X_i,\bb_0)\trans
}{\pi(k,\X_i,\bb_0)}\right\}, \\
\V_{kl}(\bb_0)&=&E\left[
\left\{\frac{I(k=l)}{\pi(k,\X_i,\bb_0)}-1\right\}\B(k,\X_i)   \B(l,\X_i)\trans\right],  \\
\C_{kl}(\bb_0)&=&E\left\{I(k=l) \frac{m(k,\X_i)^2+v(k,\X_i)}{\pi(k,\X_i,\bb_0)}-m(k,\X_i)m(l,\X_i)\right\}\\
&&+E\left([m(k,\X_i)-E\{m(k,\X_i)\}][m(l,\X_i)-E\{m(l,\X_i)\}]\right).
\ese
\end{Cor}
\begin{Rem}\label{rem:varcat}
	The variance $\bSig$ may be estimated without knowing nor
estimating $\ba$, by  approximating  the original definitions of the matrices
involved, i.e.   
$\B(\bb_0)\equiv E\{\partial \g_i(\bb_0)/\partial{\bb_0}\trans\}$ and
$\C(\bb_0)\equiv E\{\g_i(\bb_0)^{\otimes2}\}$, instead of the
expression involving $m(\cdot)$ and 
$v(\cdot)$ given in Corollary~\ref{cor:asympt}.
\end{Rem}

\begin{Cor}\label{cor:samedim}
Under the assumptions of Corollary \ref{cor:asympt},  the variance
of $\wh\btheta$ attains the semiparametric efficiency bound  
$
\bSig_{\rm eff},
$
where the $(k,l)$ entry of $\bSig_{\rm eff}$ is
\bse
\bSig_{\rm eff,k,l}=I(k=l)E\{v(k,\X)/\pi(k,\X)\}
+E([m(k,\X)-E\{m(k,\X)\}][m(l,\X)-E\{m(k,\X)\}]).
\ese
\end{Cor}

\section{Continuous treatments}\label{sec:cont}
\subsection{Balancing scores and preliminaries on estimation}

We now consider a continually valued treatment $A$, say taking values $a$
in $[0,1]$. In this case, it is 
reasonable to assume that the potential outcome $Y^a$ changes with $a$
smoothly. We write $Y^a$ as $Y(a)$ in a more conventional
notation. Note that the observed outcome for the $i$th observation,
$Y_i$, is assumed to be
$Y_i(a_i)$ when we observe $A_i=a_i$. 
We  observe a random sample $(A_i, Y_i, \X_i), i=1, \dots, n$, where $\X_i\in \mR^d$ is a
vector of pre-treatment covariates observed for all units. Following
the literature convention, we assume ignorability of the treatment assignment, in the 
sense that
$E\{Y_i(a)\mid\X_i,A_i\}=E\{Y_i(a)\mid\X_i\}$, and the generalized
propensity score is the conditional probability 
density function of the continuous treatment $A_i$ given the covariates $\X_i$:  
$\pi_0(a,\x)\equiv f_{A\mid\X}(a,\x)>\delta>0$ for all
$a\in[0,1]$ and all $\x$. We write the expected conditional
potential outcome as  
$m(a,\x)\equiv E\{Y_i(a)\mid\X_i=\x\}$.

In such case, the parameter  of interest
is the treatment response function or the dose-response
function, denoted as $\theta(a)=E\{Y_i(a)\}$ for
$a\in[0,1]$.  The  average causal effects between two treatment doses, say
$a$ and $b$ are obtained by taking their contrast
$\theta(a)-\theta(b)$. We consider a parametric working model
$\pi(a,\x,\bb)$ for the propensity score $\pi_0(a,\x)$,
where $\bb\in \mR^p$,
and consider a set of basis functions $\B(a,\x):\mR^{d+1}\to \mR^q$
aiming at spanning $m(a,\x)$. 
Thus, correctly specified situations will be such that there exists
$\bb_0$ so that
\be\label{eq:contpispec}
\pi(a,\x,\bb_0)=\pi_0(a,\x),
\ee
 and there exists $\ba$ such that
\be\label{eq:contmspec}
\ba\trans\B(a,\x)=m(a,\x),
\ee
for all $a\in[0,1]$ and all $\x$. Misspecification, i.e. situations
where one of (\ref{eq:contpispec}) and (\ref{eq:contmspec})
does not hold, will be allowed in the sequel.
 
The balancing consideration then leads us to the condition
\bse
\sumi\left[ 
\left\{\frac{K_l(A_i-a)}{\pi(a,\X_i,\bb)}-1\right\}\B(a,\X_i)
-
\left\{\frac{K_l(A_i-b)}{\pi(b,\X_i,\bb)}-1\right\}\B(b,\X_i)
\right]=\0
\ese 
for two arbitrary $a,b$ values in $[0,1]$.
Following the same considerations as in Section \ref{sec:cat},
we strengthen the above requirement and consider
the balancing equations
\be\label{eq:betacon1}
\sumi
\left\{\frac{K_l(A_i-a)}{\pi(a,\X_i,\bb)}-1\right\}\B(a,\X_i)=\0
\ee 
at all $a\in[0,1]$. Here, $K_l(\cdot)=l^{-1}K(\cdot/l)$, where
 $K(\cdot)$ is a kernel function and $l$ is a bandwidth. 
 Practically, we propose to solve (\ref{eq:betacon1}) at
a set of chosen $a$ values, typically those observed for $A_i$,
and  minimize
\be\label{eq:betaconint}
\sumj \left\| \sumi\left[
\left\{\frac{K_l(A_i-A_j)}{\pi(A_j,\X_i,\bb)}-1\right\}\B(A_j,\X_i)\right]
\right\|_2^2 \{\sumi K_l(A_i-A_j)\}
\ee
with respect to $\bb$ to get $\wh\bb$.
Once we obtain
 $\wh\bb$, we estimate the causal parameter $\theta(a)$ with an
 inverse probability weighting estimator
\be\label{eq:thetacon}
\wh\theta(a)=n^{-1}\sumi
\frac{K_h(A_i-a)Y_i}{\pi(a,\X_i,\wh\bb)},
\ee
for any $a$ within the range of observed values for $A_i$.
Here, $h$ is a bandwidth. 

\begin{Rem}\label{rem:con}
The nonparametric estimator (\ref{eq:thetacon}) can be viewed as 
an approximation of 
\bse
\frac{n^{-1}\sumi  Y_iK_h(A_i-a)/\pi(A_i,\X_i,\wh\bb)}
{n^{-1}\sumi  K_h(A_i-a)/\pi(A_i,\X_i,\wh\bb)},
\ese
which is
the
solution to 
\bse
\min_c \sumi \frac{(Y_i-c)^2K_h(A_i-a)}{\pi(A_i,\X_i,\wh\bb)}.
\ese
Thus, we can understand (\ref{eq:thetacon}) as a weighted local
constant estimator of $\theta(a)$. Similar to the generalization from
local constant to local polynomial estimators in nonparametrics, we can 
also generalize (\ref{eq:thetacon}) to more sophisticated
versions. For example, 
through obtaining $\wh c_0$ from
\bse
\min_{c_0, c_1} \sumi \frac{\{Y_i-c_0-c_1(A_i-a)\}^2K_h(A_i-a)}{\pi(A_i,\X_i,\wh\bb)},
\ese
we can obtain the weighted local linear estimator of $\theta(a)$.
\end{Rem}

\subsection{Asymptotic properties}

We now study  the limiting properties of the estimator (\ref{eq:thetacon})
using (\ref{eq:betaconint}); see Appendix \ref{app:cont} for proofs.  
Denote by $\bb^*$ the probability limit of $\wh\bb$. If model (\ref{eq:contpispec})
is correct, $\bb^*=\bb_0$, otherwise $\bb^*$ is the value that 
minimizes (\ref{eq:betaconint}) at the population level, i.e. it minimizes
\be \label{eq:obj}
E_{j}\left(\left\|E_i\left[
\left\{\frac{K_l(A_i-A_j)}{\pi(A_j,\X_i,\bb)}-1\right\}\B(A_j,\X_i)
\right]
\right\|_2^2 \{\sumi K_l(A_i-A_j)\}
\right)
\ee
with respect to $\bb$. Here $E_j$ means taking expectation of
the $j$th observation. We list the following regularity conditions. 
\begin{description}
\item[C0.] $\bb^*$ is the unique solution of 
 $E\left[
\left\{\frac{\pi_0(a,\X)}{\pi(a,\X,\bb)}-1\right\}\B(a,\X) 
\right]=\0$.

\item[C1.] The kernel function $K(\cdot)\ge0$ is bounded, 
 twice differentiable with bounded first derivative,
symmetric and has support on $(-1,1)$. It satisfies
$\int_{-1}^1K(t)dt=1$. 
 
\item[C2.] The bandwidth $l$ satisfies $nl^4\to0$ and $nl^2\to\infty$.
The bandwidth $h$ satisfies $h\to0$ and $nh\to\infty$.

\item[C3.] The basis function $\B(a,\x)$ is bounded.

\item[C4.] The propensity score $\pi(a,\x,\bb)$ is differentiable
with respect to $\bb$ and $a$, is bounded away from
zero, and its derivative with respect to $a$ is bounded.

\item[C5.] $m(a,\X_i)$ is bounded, twice differentiable with respect to $a$,
 and the first derivative is bounded.

\item[C6.]  $\sigma^2(A_i, \X_i)\equiv\var(Y_i\mid A_i, X_i)$ is bounded.
\end{description}
These are typical regularity conditions.
Similar to Condition A0 in the categorical treatment
    case, the uniqueness requirement in Condition C0 can be relaxed to
    local uniqueness. Moreover, with finite samples, C0 can be translated to:
   $\bb^*$ is the unique solution of
    $E_i\left[
\left\{\frac{K_l(A_i-A_j)}{\pi(A_j,\X_i,\bb)}-1\right\}\B(A_j,\X_i)
\right]=\0$ for $j=1,\ldots,n$, which is easier to fullfil.
       The existence of $\bb^*$ is guaranteed when the
    propensity model $\pi(a,\x,\bb)$ is correctly specified, and is a standard
    requirement when the number of equations $qn$ is not larger than
    the length of $\bb$. Thus, in the situation where we are not confident
    that a correct propensity model is used, we can always enrich the
    model to accommodate Condition C0.
We start by giving the convergence rate of $\wh\bb$.

\begin{Lem}\label{lem:betaconv}
Denote by $\bb^*$ the probability limit of $\wh\bb$. If model (\ref{eq:contpispec})
is correct, $\bb^*=\bb_0$, otherwise $\bb^*$ is the value that 
minimizes (\ref{eq:obj}).
Under regularity conditions C0 to C4, $\wh\bb-\bb^*=O_p(n^{-1/2})$.
\end{Lem}
Condition C0 is not really necessary for Lemma
  \ref{lem:betaconv}. We can redefine $\bb^*$ as the unique minimum of
  (\ref{eq:obj}) and Lemma
  \ref{lem:betaconv} still holds. 
Because the
nonparametric estimation convergence rate is slower than
$O_p(n^{-1/2})$, Lemma \ref{lem:betaconv} indicates that 
we can fix $\bb$ at
$\bb^*$ in the following analysis as long as we let $nl^4\to 0$, and the first order bias and variance property of
$\wh\theta(a)$ will not be affected. 

\begin{Th} \label{th:a}
Under regularity conditions C0 to C6, and if
(\ref{eq:contpispec}) holds, then the estimator $\wh\theta(a)$ defined
by (\ref{eq:thetacon}) has asymptotic normal distribution with
asymptotic bias and variance: 
\be
E\{\wh\theta(a)\}-\theta(a)
&=&\frac{h^2}{2}E\left[
\frac{\partial^2 \{\pi_0(a,\X_i) m(a,\X_i)\}}{\pi_0(a,\X_i)\partial
  a^2}\right]
\int
t^2K(t)dt+O(h^4+n^{-1/2}), \label{eq:bias1} \\
\var\{\wh\theta(a)\}
&=&\frac{\int K^2(t)dt}{nh}E\left\{
\frac{m^2(a,\X_i)+\sigma^2(a,\X_i)}{\pi_0(a,\X_i)}
\right\}+O(n^{-1}h^{-1/2}),\label{eq:var1}
\ee
where $\sigma^2(A_i, \X_i)=\var(Y_i \mid A_i,\X_i)$.
\end{Th}
\begin{Th}\label{th:biasvar1}
Under regularity conditions C0 to C6, and if (\ref{eq:contmspec}) holds, then the estimator $\wh\theta(a)$
defined by (\ref{eq:thetacon}) has asymptotic normal distribution with asymptotic bias and variance:
\be
E\{\wh\theta(a)\}-\theta(a)
&=&\frac{h^2}{2}E
\left[\frac{\partial^2 \{\pi_0(a,\X_i)m(a,\X_i)\}}{\pi(a,\X_i,\bb^*) \partial a^2
}\right]\int t^2K(t)dt 
+O(h^4+n^{-1/2}), \label{eq:bias2} \\
\var\{\wh\theta(a)\}
&=&\frac{\int K^2(t)dt}{nh}E\left[
\frac{\pi_0(a,\X_i)\{m^2(a,\X_i)+\sigma^2(a,\X_i)\}}{\pi^2(a,\X_i,\bb^*)}
\right]+O(n^{-1}h^{-1/2}).\label{eq:var2} 
\ee

\end{Th}
Theorems \ref{th:a} and \ref{th:biasvar1} together
reflect a robust property of the 
proposed estimator, and give equivalent results when all nuisance
models are correctly specified. Specifically, Theorem \ref{th:a} 
  describes the robustness to misspecification of the
  outcome models, in that as long as the propensity score is correctly
  specified, the estimation of the treatment response function is valid even if we do not assume a correct model for the outcome.
This is because the propensity score balances any functions of the covariates.  
Theorem \ref{th:biasvar1} allows for the misspecification of the
propensity score, with the restriction that Condition C0 needs to
hold. If we choose to ensure C0 through allowing sufficiently many model
parameters, then $\bb$ will have length $p=qn$, which practically means that the
propensity score is non-parametrically estimated. For example, 
we can let $\pi(a_j,\x)=\bb_{(j)}\trans\B(a_j,\x)$, where
$\bb_{(j)}$ has dimension $q$. Then, solving
(\ref{eq:betacon1}) for all observed $a=a_j$ corresponds to minimizing
the loss function
$$
 \sum_{i=1}^n [{K_l(A_i-a_j)}\log \{\bb_{(j)}\trans\B(a_j,\X_i)\} - \bb_{(j)}\trans\B(a_j,\X_i)],$$ for $j=1, \dots, n$.

 Finally, note here, that the dose response function $\theta(a)$ is estimated nonparametrically, and this estimation has bias of order $h^2$, although
asymptotically vanishing, and there is the usual bias-variance
trade-off. Next,
we give a result useful for inference on a causal contrast
$\theta(a)-\theta(b)$.

\begin{Th}\label{th:biasvar2}
Under regularity conditions C0 to C6, and if either
(\ref{eq:contpispec}) or (\ref{eq:contmspec}) hold, then
$\wh\theta(a)-\theta(a)$ defined by (\ref{eq:thetacon}) is
asymptotically a Gaussian process, and has asymptotic
variance-covariance:
\be\label{eq:cov}
&&\cov\{\wh\theta(a),\wh\theta(b)\}\\
&=& 
(nh)^{-1}E\int_0^1
\frac{K(t)K(t+c)\{m^2(a,\X_i)+\sigma^2(a,\X_i)\}}
{\pi(a,\X_i,\bb^*)\pi(b,\X_i,\bb^*)}\pi_0(a,\X_i)dt\n\\ 
&&+n^{-1}E\int_0^1K(t)K(t+c)
\left\{2m(a,\X_i)m'_a(a,\X_i) \pi_0(a,\X_i)+m^2(a,\X_i) \pi_{0a}'(a,\X_i)\right.\n\\ 
&&\left.+2\sigma(a,\X_i)\sigma'_a(a,\X_i) \pi_0(a,\X_i)+\sigma^2(a,\X_i) \pi_{0a}'(a,\X_i)\right\}t/\{\pi(a,\X_i,\bb^*)\pi(b,\X_i,\bb^*)\} 
dt\n\\
&&-n^{-1}\theta(a)\theta(b)
+O(n^{-1}h+h^{-1}n^{-3/2}),\n
\ee
where $c\equiv(a-b)/h$.
\end{Th}

Note that when $c\notin(-2,1)$, $K(t)K(t+c)=0$ for all $t$. Therefore, 
the covariance has order $O(n^{-1})$ if
$c\notin(-2,1)$ and $O\{(nh)^{-1}\}$ otherwise. Thus, comparing the term of
order $O\{(nh)^{-1}\}$ in the covariance 
 in Theorem \ref{th:biasvar2}
with the terms of the same order for
the variances in Theorems \ref{th:a} and \ref{th:biasvar1}, we
see that when $a$ and $b$ are close to each other relative to $h$,
the variance 
of the contrast $\wh\theta(a)-\wh\theta(b)$ is close to zero. On the
contrary, when $a$ and $b$ are far apart, then the variance of the
contrast is dominated by the variance of $\wh\theta(a)$ and
$\wh\theta(b)$. 

Theorems \ref{th:a}, \ref{th:biasvar1} and \ref{th:biasvar2}
provide theoretical properties of the leading orders of the bias,  variance
and covariance properties of the nonparametric estimators. In large
samples, these results can be used to perform
inference. Practically, unlike for parameter estimation, because the next
order of the nonparametric analysis is only slightly smaller than
the leading order,  inference based on these results is often not sufficiently
precise. This phenomenon has been observed in many nonparametric or
even semiparametric problems including quantile regression, survival
analysis, etc., and bootstrap is often used instead.

\section{Simulation Experiments}\label{sec:simu}

\subsection{Categorical treatment}

To investigate the finite sample performance of our method for the 
categorical treatment case, we performed a first simulation study. 
We generate a five dimensional covariate vector $\X$, where 
$X_1=1$, and $X_2$ to $X_5$ are generated independently from 
a normal distribution with mean 3 and variance 4. 
We set $K=3$ and 
the propensity score 
$\pi_0(k,x)=\exp(\x\trans\bb_k)/\{1+\sum_{k=0}^2 
\exp(\x\trans\bb_k)\}$ for $k=0,1,2$, and let 
$\pi_0(3,x)=1-\sum_{k=0}^2\pi_0(k,x)$. 
Here, 
$\bb_0=(0,-0.2475,-0.275,0.1875,0.075)\trans$, 
$\bb_1=(0,-0.165,-0.15,0.125,0.05)\trans$, and 
$\bb_2=\0$. 
We set $m(k,\x)=\ba_k\trans\x$, where 
$\ba_0=(200,0,13.7,13.7,13.7)\trans$, and $\ba_1$ to $\ba_3$ are set to be 
$(200,27.4,13.7,13.7,13.7)\trans$. 
We generated $Y_i^k$'s by adding 
a standard normal random noise to the true mean $m(k,\x_i)$. 

In 
implementing the estimators, in addition to the ideal case where both 
the $\pi(\cdot)$ model and the basis for the $m(\cdot)$ model are 
correct, we also experiment with incorrectly specified models. In 
misspecifying the $\pi(\cdot)$ models, we replace $X_1$ with $e^{X_1}$, 
$X_2$ with $X_1X_2$, $X_3$ with $X_1^2X_3$, $X_4$ with $X_1+X_4$ and 
$X_5$ with $X_5\sin(X_5)^2$. 
In misspecifying the $m(\cdot)$ models, we 
replace $X_1$ with $X_1^2$, $X_2$ with $X_1X_2$, $X_3$ with $X_2X_3^2$
and $X_4$ with $(X_4-3)^3+3$. 
We investigate four 
different scenarios, when both models are correct, when the $\pi(\cdot)$
model is misspecified, when the $m(\cdot)$ model is misspecified and when
both models are misspecified. 
Note that our design is such that correctly specifying the basis for 
$m(\cdot)$ corresponds to balancing the first moments of the 
covariates. 
For comparison, we also implemented the inverse probability 
weighting estimators (IPW) using maximum likelihood for the estimation
of the propensity score, and its double robust augmented version using
both the correct propensity score and outcome models; for the latter
we use the R-package \texttt{PSweight} (\citeauthor{PSweight}
\citeyear{PSweight}). 
The results over 1000 replicates are 
displayed in Tables \ref{tab:cate1}-\ref{tab:cate3} (see Appendix \ref{app:tab})
for different sample sizes, where for each causal contrast 
$\theta_k-\theta_0$, $k=1,2,3$, we provide bias, 
standard deviation, mean squared errors (MSE) as well as 
average estimated standard deviation, and 
empirical coverage of the resulting 95\% confidence interval. See 
Remark \ref{rem:varcat} for how the inference is carried out. 

\begin{figure}[h!]
\centering 
\includegraphics[width=0.8\linewidth]{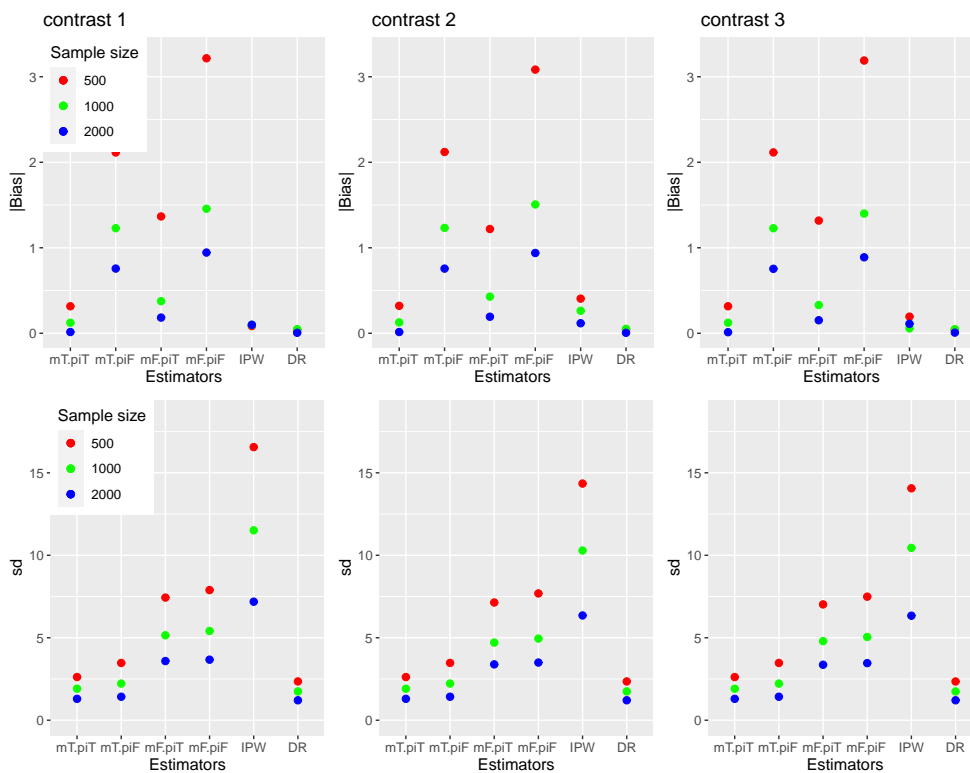}
\caption{Absolute bias and sd for the three contrasts
  $\theta_j-\theta_0$, $j=1,2,3$, over 1000 replicates for the six
  estimators: $m,\pi$ correct (mT.piT), $m$ correct (mT.piF), $\pi$
  correct (mF.piT), $m,\pi$ misspecified (mF.piF), IPW and augmented IPW (DR), and three sample sizes.
}
\label{fig:msecat}
\end{figure}

Biases and standard deviations are also displayed graphically in 
Figure \ref{fig:msecat}. These numerical experiments confirm
the theoretical 
robustness properties in 
the sense that much smaller biases are observed when at least one of 
the models is correctly specified compared to when both models 
$\pi(\cdot)$  
and $m(\cdot)$ are misspecified. Increasing sample sizes improves
biases and variances as expected, except when all models are
misspecified. Moreover, compared to the maximum likelihood based  
inverse probability weighting method (ML-IPW), our estimator yields
lower variance, and its MSE is smaller even when both 
models are misspecified. The classical augmented IPW (DR) should be
considered as a benchmark, since in contrast with our estimator which
only fits the propensity score, DR fits all models. Fitting the
outcome models is, however, arguably not desirable \citep{Rubin:2007},
and it appears to yield lower finite sample bias and variance in the
cases considered. The relative efficiency of our estimator compared to
DR improves with increasing sample sizes although slowly.   
Empirical 
coverages match the nominal level of 95\%, and this gets better with
increasing sample size, except for when all models are misspecified as
expected from theory.

\subsection{Continuous treatments}

To assess the  performance of the proposed methods under continuous
treatment, we experiment with both
linear and nonlinear outcome models. In the nonlinear design, we 
 generate a five dimensional covariate vector $\X$, where $X_{1}=1$
and  $(X_2,X_3,X_4,X_5)\trans$ follows a multivariate standard normal
distribution. Thus, these covariates have mean
zero, variance 1 and are independent of each other.
The true propensity score function is
\bse
\pi_0(a,\x)=\frac{\Gamma(15)}{\Gamma[15\lambda(\x)]\Gamma[15\{1-\lambda(\x)\}]}\Big(\frac{a}{20}\Big)^{15\lambda(\x)-1}\Big(1-\frac{a}{20}\Big)^{15\{1-\lambda(\x)\}-1}\frac{1}{20}.
\ese
Note that this is the probability density function of $A$ when $A/20$
follows a beta distribution with parameters
$15\lambda(\x)$ and $15\{1-\lambda(\x)\}$, 
where 
$\logit\{\lambda(\x)\}=(-0.8,0.1,0.1,-0.1,0.2)\x$. We further
generate the response $Y$  from a Bernoulli distribution with
probability $m_1(A,\X)\equiv\text{expit}\{\mu(A,\X)\}$,
where $\mu(a,\x)= (1,0.2,0.2,0.3,-0.1)\x + a(0.1,-0.1,0,0.1,0)\x -0.13^3a^3$.
This simulation design is identical to that of \cite{kennedyetal:2017}.
In the linear design, the response is generated from a
normal distribution with mean
$m_2(A,\X)$ and variance 0.16,  where $m_2(a,\x)=\{\mu(a,\x)+15\}/20$.

Two different types of IPW estimators are implemented in both linear and
nonlinear outcome cases, respectively a maximum likelihood based
inverse probability weighting 
estimator and the proposed robust balancing estimator. For the former,
we used a maximum likelihood approach
to estimate the parameter of the 
propensity score. For the balancing
estimator, (\ref{eq:betaconint}) is minimized 
where the bandwidth $l$ was set to
$3n^{-1/3}$. In the nonparametric estimation of $\theta(a)$
in (\ref{eq:thetacon}),
both the local constant and local linear estimators given in Remark
\ref{rem:con} are implemented
and $h$ was selected by
the leave-one-out cross-validation and the one-sided cross-validation
\citep{oscv}.  
For comparison, the inverse probability weighted and the doubly robust
estimator given in \cite{kennedyetal:2017} are also implemented using
the R-package \texttt{npcausal} (github.com/ehkennedy/npcausal).

For the linear outcome case, 
 the estimators are assessed in four different scenarios where both
models are correct or either of the models is misspecified. We use the
basis of $\mu(a,\x)$ as basis of the outcome model.
In misspecifying either the $\pi(\cdot)$ or $m(\cdot)$ model, we
replaced the covariates with $\x^*$ as in
\cite{kang2007demystifying}, with
\bse
    \x^* = \left\{1, e^{x_2/2}, \frac{x_3}{1+\exp(x_2)}+10, (x_2x_4/25+0.6)^3, (x_3+x_5+20)^2 \right\}\trans.
\ese
In addition, the misspecified $m_i(\cdot) \, (i=1,2)$ has no cubic
term of $a$ in its bases. We in fact used the same construction for
the nonlinear outcome model. However, we point out that this leads to
the scenario that the outcome model basis is never correctly
specified, while the propensity score model is either correct or incorrect.

\begin{figure}[h!]
    \centering 
    \includegraphics[width=0.8\textwidth]{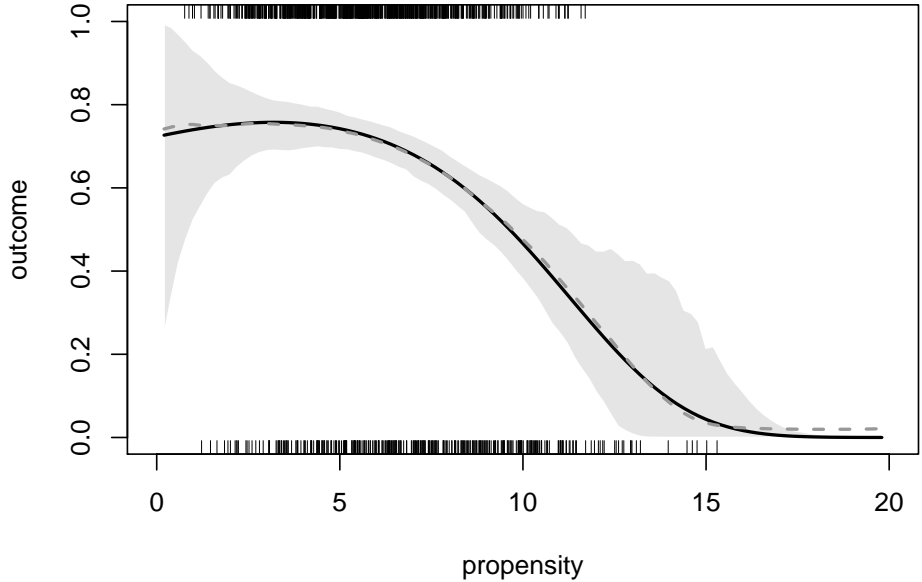}
    \caption{Simulation in the continuous nonlinear outcome case. Rug: One
      simulated data set with $n$=1000; Solid: True 
      outcome; Dotted: Mean of the estimates, i.e., $\frac{1}{T}\sum_{t=1}^T\hat\theta_t(a)$, using local constant estimation and CV, and $T=1000$; Filled curves: 5\% and 95\% quantiles of $\hat\theta_t(a)$.}
    \label{fig:cont}
\end{figure}

We generated the simulated data with sample sizes $n=500, 1000, 2000$
and the result is based on $1000$ replicates. Figure \ref{fig:cont}
illustrates the simulated data with the nonlinear outcome model and the
empirical coverage of the proposed estimator under $n=1000$. We
assessed the performance of each estimator by calculating the integrated
absolute bias and the integrated root-mean-squared error (RMSE), where
\bse
\text{bias} &=& \int_{\mathcal{A}^*}\left|E\{\hat\theta(a)\}-\theta(a)\right|f_A(a)da, \\
\text{RMSE} &=& \int_{\mathcal{A}^*}E\left[\{\wh\theta(a)-\theta(a)\}^2\right]^{1/2}f_A(a)da,
\ese
where $\mathcal{A}^*$ is a trimmed support of $A$ which excludes 10\%
mass on the boundaries.

The results are given in Tables \ref{tab:nonlinear} and
\ref{tab:linear} (Appendix \ref{app:tab}). The integrated absolute
bias and the integrated RMSE 
are numerically calculated and presented with the integrated RMSE in
parentheses. For ease of presentation, both measures are multiplied
by 100. These results confirm that  the proposed estimator is
robust. In addition, as seen in Table
\ref{tab:nonlinear}, we find that our estimator shows robust
performance even under the nonlinear outcome design where
(\ref{eq:contmspec}) does not hold, which means that none of the four
cases used the true basis of the outcome model. Among the balancing
estimators, the variant using local linear fit and one-sided CV seems
to perform best in terms of bias and RMSE when both all nuisance
models are correctly specified. The 
balancing method has also both lower bias and RMSE than the IPW
estimators. We note that the bias is most sensitive to specification
of 
the propensity score model. In all cases, the proposed estimator
outperforms the estimator by \cite{kennedyetal:2017} in terms of bias,
although RMSE Kennedy's double robust estimator has lowest RMSE. Here,
as for the categorical case, this estimator can be considered a
benchmark since it fits also outcome models in contrast with the
introduced balancing estimators.  

\section{Effect of BMI on self reported health decline}\label{sec:app}

As a case study, we investigate the effect of Body Mass Index (BMI) on
self reported health (SRH) decline. This analysis is based on data
from the Survey of Health, Aging and Retirement in Europe
(SHARE). This is an interview based longitudinal survey of individuals
of age 50 years or older \citep{BS:13}. Here we use  data on women
from three countries (Sweden, Netherland, Italy) that participate in
waves 1 and 5 of the SHARE study. Wave 1 data collected in 2004 serve as the
baseline, and individuals are followed up at wave 5, collected in
2013. We are interested in estimating the average causal effect of BMI
(a continuous valued treatment with range 15.62-49.60 in the data) on
SRH  decline between baseline and follow-up. SRH is measured by
asking the question ``Would you say your health is: excellent, very
good, good, fair or poor?'' Despite its unspecific
  nature, SRH has been found to predict mortality well in many studies
  \citep{idleretal:97}, and is thus considered as an important health
  indicator.
SRH decline is here defined as a binary
variable which, for the respondents reporting ``excellent, very good,
or good health'' at baseline, will take value one if they changed their
answer to ``fair or poor health'' at follow-up, and 0 otherwise. The
resulting sample of complete cases consists of 1530 participants. In \cite{genbacketal:08}, predictors of SRH decline were investigated using logistic
regression, and it was found that BMI measured at baseline was a
significant (5\% level) predictor of SHR decline. Here we aim at
sharpening this analysis and study whether there is evidence that BMI
is a causal agent of SRH decline by using the introduced covariate
balancing procedure for causal inference. The covariates observed at
baseline that we use for balancing are age (years), whether the
participant responded to the SRH question at the beginning of the
interview (or the end), socio-economic variables (education level,
make ends meet easily), cognitive function variables (numeracy test,
date orientation question), health variables (number of chronic
diseases, number of mobility problems, depression measure, maximum
grip strength, limitation in normal activities), and lifestyle
variables (smoking habits, alcohol usage, physical activities). 
We refer to \cite{genbacketal:08} for a detailed description of these
covariates.
Encouraged by \cite{overweight:97} 
and \cite{ng2016novel},
our analysis is based on the following model for $A=
(\text{BMI}-15)/40 $ given the covariate vector $\x$:
\bse
\pi_0(a,\x) &=& \frac{\Gamma(\phi)}{\Gamma[\phi\lambda(\x)]\Gamma[\phi\{1-\lambda(\x)\}]}a^{\phi\lambda(\x)-1}(1-a)^{\phi\{1-\lambda(\x)\}-1} , \\
\logit\{\lambda(\x)\} &=& \bg\trans\x, \\
\bb &=& (\bg, \phi) .
\ese
The basis functions for the outcome model are chosen to be
$
\B(a,\x) = (\x,a,a^2,a^3).
$
A value for
$\bb^{(0)}=(\bg^{(0)},\phi^{(0)})$ is obtained by the maximum
likelihood estimation and used as the starting value for solving the
balancing equations (\ref{eq:betaconint}), with the bandwidth
$l=6n^{-1/3}$.
For nonparametric estimation of $\theta(a)$
in (\ref{eq:thetacon}), the local constant estimator given in Remark
\ref{rem:con} is used for simplicity, 
where $h$ was selected by one-sided cross-validation \citep{oscv}.
\begin{figure}[h!]
    \centering 
    \includegraphics[width=0.85\textwidth]{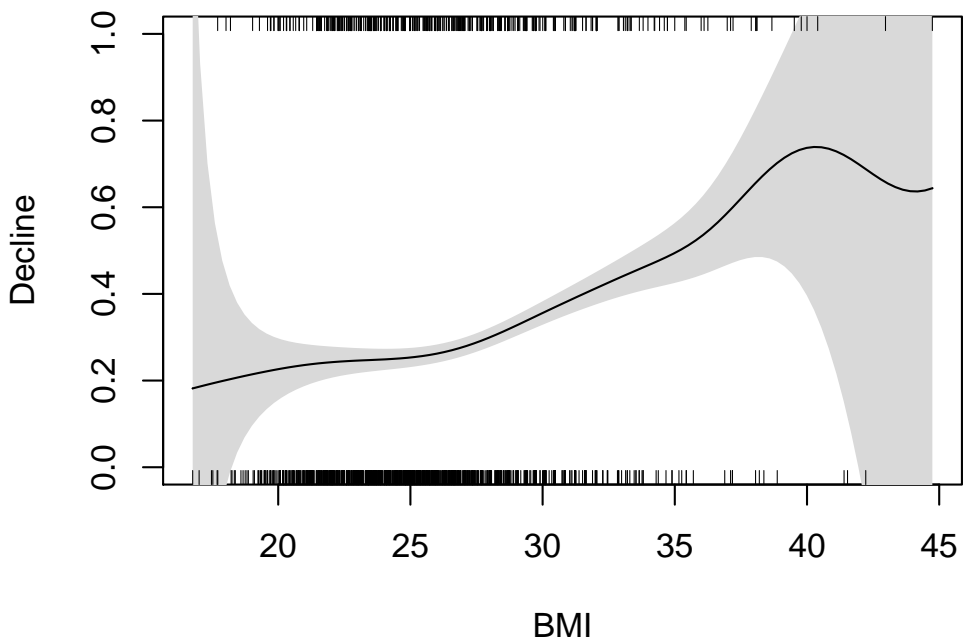}
    \caption{Effect of BMI on SRH decline.
Rug plot: the observations;
solid line: the estimated average treatment effect curve;
filled gray curve: the estimated pointwise confidence band.} 
\label{fig:real}
\end{figure}

Figure \ref{fig:real} displays the estimated effect curve of BMI on
SRH decline. Confidence bands are obtained using the variance
estimates described in Appendix \ref{app:varest}. Overall, we observe a
nonlinear effect curve. Specifically, we observe that
BMI has no significant effect for
values of BMI considered as normal (i.e. below 25) in that the
confidence band of the 
probability of decline contains the flat line. However in the
range of BMIs considered as
overweight (BMI larger than 25), an increase in the probability of
SRH decline  is observed, reflecting the causal effect of 
the increase of BMI on the probability of SRH decline. The
causal interpretation of this effect relies on the assumptions made. 
Mainly that all confounders have been observed, and
that a well defined intervention on BMI corresponds to the effect
measured \citep{Hernan2008}. Nevertheless, the results are in 
line with earlier studies pointing at a wide range of health risks 
from overweight and obesity (Afshin et al. \citeyear{overweight:97}).

\section{Discussion}\label{sec:disc}
We have introduced novel robust estimation and inference tools
for multi-level treatments. For continuous treatments our
proposal together with that of \cite{kennedyetal:2017} are, to the
best of our knowledge, the only  robust methods which model the
causal dose-response curve nonparametrically. Our
  results expand the recent important developments given by
  \cite{Fanetal:2018}.  
 For both the categorical and continuous treatment cases, we achieve robustness by balancing
basis functions for the outcome models when fitting a generalized
propensity score model which is either correct or sufficiently rich.
While the estimator proposed is locally efficient for
  the categorical case, asymptotic efficiency is not relevant for the
  continuous case where the parameter of interest is a function of the
  dose and is estimated non-parametrically.  

The proposal differs from earlier double robust
estimation in that it does not need outcome models to be fitted.  
This is an advantage when outcome is not observed
  at the design stage of the study. Indeed, it is argued that
  observational studies should be designed without using observed
  outcomes even if available in order to mimic the ``objectivity'' of the
  designs of randomized trials; see \citep{Rubin:2007} for a detail discussion. Our
  simulation results indicate that this is done at a cost in finite
  sample performance.
Our work is somewhat in contrast to the widespread
practice of using simple (e.g. linear or logistic linear) models for
the propensity score with matching estimators assuming that balance in
the joint distribution of the covariates is achieved 
\citep[e.g.,][]{IW:10,RubinThomas:00}. 
However, balancing the joint
distribution is not necessary, and in exchange, more elaborate requirements are
on the propensity score.
 From the results presented herein, it becomes 
transparent which functions of the covariates are sufficient to balance for in
order to both obtain consistency and, in the categorical treatment
case, local efficiency. 
    
In high-dimensional settings ($d\approx n$), it has recently been
shown that bias due to regularization in estimating correctly
specified linear outcome models can be corrected by using relevant
weights which are not necessarily based on the true propensity score
\cite[][]{atheyetal:18}; see also, e.g., \cite{Farrell:15} and
\cite{Dukesetal:20} for double robust estimation with many
covariates. An interesting future direction of research is whether one
can generalize the results presented herein to high-dimensional
situations, balancing many basis functions for the outcome models by
using, e.g., regularized GMM techniques \citep{alex2018highdimensional}.

\section*{Acknowledgements}
The Marianne and Marcus Wallenberg Foundation and the Swedish Research
Council are acknowledged for their financial support. 

\baselineskip=14pt
\bibliographystyle{agsm}
\bibliographystyle{chicago} 
\bibliography{catcon39}
	

\newpage
\baselineskip=28pt

\section*{Appendix}

\renewcommand{\thesection}{A.\arabic{section}}
\renewcommand{\thesubsection}{A.\arabic{subsection}}

\subsection{Categorical treatment: derivations}\label{app:cat}
\subsubsection{Asymptotic distribution and variance of $\wh\theta_k$'s}
Let 
\bse
\f_{ki}(\bb)\equiv\left\{\frac{I(A_i=k)}{\pi(k,\X_i,\bb)}-1\right\}\B(k,\X_i),
\ese
 $\f_i(\bb)\equiv\{\f_{1i}(\bb)\trans,
 \dots,\f_{Ki}(\bb)\trans\}\trans$,
 $\V(\bb)\equiv E\{ \f_i(\bb)\f_i(\bb)\trans\}$,
 $\wh\V(\bb)\equiv n^{-1}\sumi \f_i(\bb)\f_i(\bb)\trans$,
$\A(\bb)\equiv E\left\{ \partial\f_i(\bb)/\partial\bb\trans\right\}$
and
$\wh\A(\bb)\equiv n^{-1}\sumi\partial\f_i(\bb)/\partial\bb\trans$.

\begin{Lem}\label{lem:gmm}
Under regularity conditions A0, A1, A2 and A3, the GMM
 estimator $\wh\bb$ obtained by minimizing
$\{\sumi\f_i(\bb)\}\wh\V(\bb)^{-1}\{\sumi\f_i(\bb)\}$, is such that
\bse
n^{1/2}(\wh\bb-\bb^*)=-\{\A(\bb^*)\trans\V(\bb^*)^{-1}\A(\bb^*)\}^{-1}
\A(\bb^*)\trans\V(\bb^*)^{-1}\{n^{-1/2}\sumi\f_i(\bb^*)\}+O_p(n^{-1/2}). 
\ese
When (\ref{eq:pispec}) holds $\bb^*=\bb_0$.
\end{Lem}

\begin{proof}
The GMM
 estimator $\wh\bb$ is obtained by minimizing
$\{\sumi\f_i(\bb)\}\wh\V(\bb)^{-1}\{\sumi\f_i(\bb)\}$. 
This entails
\bse
\0&=&
\wh\A(\wh\bb)\trans\wh\V(\wh\bb)^{-1}\{n^{-1/2}\sumi\f_i(\wh\bb)\}
+\frac{n^{1/2}}{2}[\{\frac{1}{n}\sumi\f_i(\bb)\}\frac{\partial\{\wh\V(\bb)^{-1}\}}{\partial\beta_k}\{\frac{1}{n}\sumi\f_i(\bb)\}]_{k=1}^p\\
&=&\wh\A(\bb)\trans\wh\V(\bb)^{-1}\{n^{-1/2}\sumi\f_i(\wh\bb)\}+O_p(n^{-1/2})\\
&=&\A(\bb^*)\trans\V(\bb^*)^{-1}\{n^{-1/2}\sumi\f_i(\bb^*)\}
+\A(\bb^*)\trans\V(\bb^*)^{-1}\A(\bb^*)n^{1/2}(\wh\bb-\bb)
+O_p(n^{-1/2}),
\ese
hence
\bse
n^{1/2}(\wh\bb-\bb^*)=-\{\A(\bb^*)\trans\V(\bb^*)^{-1}\A(\bb^*)\}^{-1}
\A(\bb^*)\trans\V(\bb^*)^{-1}\{n^{-1/2}\sumi\f_i(\bb^*)\}+O_p(n^{-1/2}).
\ese
\end{proof}

\begin{proof}[Proof of Theorem \ref{thm:asympt}]
Using Lemma \ref{lem:gmm} we can write
\bse
n^{-1/2}(\wh\btheta-\btheta)
&=&n^{-1/2}\sumi\g_i(\wh\bb)\\
&=&n^{-1/2}\sumi\{\g_i(\wh\bb)-\g_i(\bb^*)\}
+n^{-1/2}\sumi\g_i(\bb^*)\\
&=&-\B(\bb^*)
\{\A(\bb^*)\trans\V(\bb^*)^{-1}\A(\bb^*)\}^{-1}
\A(\bb^*)\trans\V(\bb^*)^{-1}\{n^{-1/2}\sumi\f_i(\bb^*)\}\\
&&+n^{-1/2}\sumi\g_i(\bb^*)
+O_p(n^{-1/2}).
\ese
When either (\ref{eq:pispec}) and/or (\ref{eq:mspec}) hold, we already
know that $E\{\g_i(\bb^*)\}=\0$. Thus, under regularity conditions, 
$\sqrt{n}(\wh\btheta-\btheta)$ has asymptotic normal distribution with
mean zero and variance 
\bse
\bSig&=&
\var\left[
-\B(\bb^*)
\{\A(\bb^*)\trans\V(\bb^*)^{-1}\A(\bb^*)\}^{-1}
\A(\bb^*)\trans\V(\bb^*)^{-1}\f_i(\bb^*)
+\g_i(\bb^*)\right]\\
&=&\B(\bb^*)
\{\A(\bb^*)\trans\V(\bb^*)^{-1}\A(\bb^*)\}^{-1}\B(\bb^*)\trans
+\C(\bb^*)\\
&&-\B(\bb^*)\{\A(\bb^*)\trans\V(\bb^*)^{-1}\A(\bb^*)\}^{-1}
\A(\bb^*)\trans\V(\bb^*)^{-1}\D(\bb^*)\\
&&-\D(\bb^*)\trans
[\B(\bb^*)\{\A(\bb^*)\trans\V(\bb^*)^{-1}\A(\bb^*)\}^{-1}
\A(\bb^*)\trans\V(\bb^*)^{-1}]\trans,
\ese
where
$\C(\bb^*)\equiv E\{\g_i(\bb^*)^{\otimes2}\}$ and
$\D(\bb^*)\equiv E\{\f_i(\bb^*)\g_i(\bb^*)\trans\}$.
\end{proof}

\begin{proof}[Proof of Corollary \ref{cor:asympt}]
When all models are correctly specified, i.e. (\ref{eq:pispec}-\ref{eq:mspec}) hold, we have $\bb^*=\bb_0$. Then
\bse
\A_k(\bb_0)&=&E\left\{
-\frac{\B(k,\X_i)\pi'_\bb(k,\X_i,\bb_0)\trans
}{\pi(k,\X_i,\bb_0)}\right\}, \\
\B_{k}(\bb_0)&=&E\left\{
-\frac{m(k,\X_i)\pi'_\bb(k,\X_i,\bb_0)\trans
}{\pi(k,\X_i,\bb_0)}\right\}, \\
\V_{kl}(\bb_0)&=&E\left[
\left\{\frac{I(k=l)}{\pi(k,\X_i,\bb_0)}-1\right\}\B(k,\X_i)   \B(l,\X_i)\trans\right],  \\
\C_{kl}(\bb_0)&=&E\left\{I(k=l) \frac{m(k,\X_i)^2+v(k,\X_i)}{\pi(k,\X_i,\bb_0)}-m(k,\X_i)m(l,\X_i)\right\}\\
&&+E\left([m(k,\X_i)-E\{m(k,\X_i)\}][m(l,\X_i)-E\{m(l,\X_i)\}]\right),\\
\D_{kl}(\bb_0)&=&E\left[\left\{\frac{I(k=l)}{\pi(k,\X_i,\bb_0)}
-1\right\}\B(k,\X_i) m(l,\X_i) 
\right],
\ese
and
$\A(\bb_0)=\{\A_1(\bb_0)\trans, \dots, \A_K(\bb_0)\trans\}\trans$, 
$\B(\bb_0)=\{\B_1(\bb_0) \trans, \dots, \B_K(\bb_0)\trans\}\trans$, 
$\V(\bb_0)=\{\V_{kl}(\bb_0)\}_{k,l=1}^K$,
$\C(\bb_0)=\{\C_{kl}(\bb_0)\}_{k,l=1}^K$,
$\D(\bb_0)=\{\D_{kl}(\bb_0)\}_{k,l=1}^K$. 

Note that $\ba\trans\A_k(\bb_0)=\B_k(\bb_0)$ and
$\V_{kl}(\bb_0)\ba=\D_{kl}(\bb_0)$, so
$(\I_{K+1}\otimes\ba\trans)\A(\bb_0)=\B(\bb_0)$ and $\V(\bb_0)(\1_{K,K}\otimes\ba)=\D(\bb^*)$.
Thus, 
$\bSig=\C(\bb_0)-\B(\bb_0)
\{\A(\bb_0)\trans\V(\bb_0)^{-1}\A(\bb_0)\}^{-1}\B(\bb_0)\trans$.
\end{proof}

\begin{proof}[Proof of Corollary \ref{cor:samedim}]
	Here we have set the dimension of 
$\f_i(\bb)$ to be the same as the dimension of $\bb$ hence we can solve
$\sum\f_i(\bb)=\0$ directly. As a consequence, we can write
\bse
\0&=&n^{-1/2}\sumi\f_i(\wh\bb)+O_p(n^{-1/2})\\
&=&n^{-1/2}\sumi\f_i(\bb^*)
+\A(\bb^*)n^{1/2}(\wh\bb-\bb)
+O_p(n^{-1/2}),
\ese
hence
\bse
n^{1/2}(\wh\bb-\bb)=-\A(\bb^*)^{-1}
\{n^{-1/2}\sumi\f_i(\bb^*)\}+O_p(n^{-1/2}).
\ese
This leads to
\bse
n^{-1/2}(\wh\btheta-\btheta)
&=&n^{-1/2}\sumi\g_i(\wh\bb)\\
&=&n^{-1/2}\sumi\{\g_i(\wh\bb)-\g_i(\bb_0)\}
+n^{-1/2}\sumi\g_i(\bb_0)\\
&=&-\B(\bb_0)
\A(\bb_0)^{-1}
\{n^{-1/2}\sumi\f_i(\bb_0)\}
+n^{-1/2}\sumi\g_i(\bb_0)
+O_p(n^{-1/2})\\
&=&-(\I_{K+1}\otimes\ba\trans)
\{n^{-1/2}\sumi\f_i(\bb_0)\}
+n^{-1/2}\sumi\g_i(\bb_0)
+O_p(n^{-1/2}).
\ese
Thus, 
$\sqrt{n}(\wh\btheta-\btheta)$ has asymptotic normal distribution with mean zero and variance
\bse
\bSig&=&
\var\left\{\g_i(\bb_0)
-(\I_{K+1}\otimes\ba\trans)
\f_i(\bb_0)\right\}\\
&=&\var\left(\left[\begin{array}{c}
\frac{I(A=0)Y}{\pi(0,\X)}-E\{m(0,\X)\}\\
\vdots\\
\frac{I(A=K)Y}{\pi(K,\X)}-E\{m(K,\X)\}
\end{array}\right]-(\I_{K+1}\otimes\ba\trans)\left[
\begin{array}{c}
\left\{\frac{I(A=0)}{\pi(0,\X)}-1\right\}\B(0,\X)\\
\vdots\\
\left\{\frac{I(A=k)}{\pi(K,\X)}-1\right\}\B(K,\X)
\end{array}\right]\right)\\
&=&\var\left(\left[\begin{array}{c}
\frac{I(A=0)Y}{\pi(0,\X)}-E\{m(0,\X)\}\\
\vdots\\
\frac{I(A=K)Y}{\pi(K,\X)}-E\{m(K,\X)\}
\end{array}\right]-\left[
\begin{array}{c}
\left\{\frac{I(A=0)}{\pi(0,\X)}-1\right\}m(0,\X)\\
\vdots\\
\left\{\frac{I(A=k)}{\pi(K,\X)}-1\right\}m(K,\X)
\end{array}\right]\right)\\
&=& \var\left\{\begin{array}{c}
\frac{I(A=0)\{Y-m(0,\X)\}}{\pi(0,\X)}+m(0,\X)-E\{m(0,\X)\}\\
\frac{I(A=1)\{Y-m(1,\X)\}}{\pi(1,\X)}+m(1,\X)-E\{m(1,\X)\}\\
\vdots\\
\frac{I(A=K)\{Y-m(K,X)\}}{\pi(K,\X)}+m(K,\X)-E\{m(K,\X)\}
\end{array}\right\},
\ese
i.e., the $(k,l)$ entry of $\bSig$ is
\bse
\bSig_{kl}=I(k=l)E\left\{\frac{v(k,\X)}{\pi(k,\X)}\right\}
+
E([m(k,\X)-E\{m(k,\X)\}][m(l,\X)-E\{m(l,\X)\}]).
\ese
Compared to the semiparametric efficiency bound obtained in Section \ref{sec:effbound} below, we see that the estimator is asymptotically
efficient. 
\end{proof}

\subsubsection{Semiparametric efficiency bound}\label{sec:effbound}
The original model can be written in general as
\be\label{eq:model}
f_{\X,A,Y}(\x,a,y)=f_\X(\x) \prod_{k=0}^K [\pi(k,\x)
f_{\epsilon\mid(A,\X)}\{y-m(k,\x),k,\x\}]^{I(a=k)},
\ee
where $\pi(k,\x)$ satisfies $0<\pi(k,\x)<1$,
$\sum_{k=0}^K\pi(k,\x)=1$ and
$f_{\epsilon\mid(A, \X)}\{y-m(k,\x),k,\x\}$ satisfies 
$\int f_{\epsilon\mid(A, \X)}(\epsilon,k,\x)d\epsilon=1$ and
$\int \epsilon
f_{\epsilon\mid(A,\X)}(\epsilon,k,\x)d\epsilon=0$ for all $k=0, \dots, K$.
The parameter of interest is 
$\btheta=(\theta_1, \dots, \theta_K)\trans$, where
$\theta_k=E\{m(k,\X)\}$. Here, we sometimes write
$\epsilon=y-m(a,\x)$ for convenience.
Consider an arbitrary parametric submodel
\bse
f_{\X,A,Y}(\x,a,y,\bd)=f_\X(\x,\bzeta) \prod_{k=0}^K [\pi(k,\x,\bb)
f_{\epsilon\mid(A,\X)}\{y-m(k,\x,\ba),k,\x,\bg\}]^{I(a=k)},
\ese
where $\bd=(\bzeta\trans,\bb\trans,\ba\trans,\bg\trans)\trans$.
We get the score function
$\bS_\bd=(\bS_\bzeta\trans,\bS_\bb\trans,\bS_\ba\trans,\bS_\bg\trans)\trans$,
where
\bse
\bS_\bzeta&=&\frac{\partial f_\X(\x,\bzeta)/\partial\bzeta}{f_\X(\x,\bzeta)},\\
\bS_\bb&=&\sum_{k=0}^K\left\{I(A=k)\frac{\partial\pi(k,\x,\bb)/\partial\bb}{\pi(k,\x,\bb)}\right\},\\
\bS_\ba&=&\sum_{k=0}^KI(A=k)\left[-\frac{\partial m(k,\x,\ba)}{\partial\ba}\frac{\partial
f_{\epsilon\mid(A,\X)}\{y-m(k,\x,\ba),k,\x,\bg\}
/\partial\{y-m(k,\x,\ba)\}}{f_{\epsilon\mid(A,\X)}\{y-m(k,\x,\ba),k,\x,\bg\}}\right],\\
\bS_\bg&=&\sum_{k=0}^K\left[I(A=k)\frac{\partial f_{\epsilon\mid(A,\X)}\{y-m(k,\x,\ba),k,\x,\bg\}
/\partial\bg}{f_{\epsilon\mid(A,\X)}\{y-m(k,\x,\ba),k,\x,\bg\}}\right].
\ese
The tangent space of (\ref{eq:model}) is $\calT=\calT_\bzeta+\calT_\bb+\calT_\ba+\calT_\bg$,
where
\bse
\calT_\bzeta&=&[\a(\X): E\{\a(\X)\}=\0],\\
\calT_\bb&=&[\a(A,\X): \sum_{k=0}^K\a(k,\x)\pi(k,\x)=\0],\\
\calT_\ba&=&\left[\a(A,\X) \frac {f_{\epsilon\mid(A,\X)}'\{Y-m(A,\X),A,\X\}}
{f_{\epsilon\mid(A,\X)}\{Y-m(A,\X),A,\X\}}:\forall \a(A,\X)\right],\\
\calT_\bg&=&[\a(\epsilon,A,\X):E\{\a(\epsilon,A,\X)\mid A,\X\}=\0, E\{\epsilon\a(\epsilon,A,\X)\mid A,\X\}=\0].
\ese
The parameter of interest in the submodel is
\bse
\btheta(\bzeta,\bb,\ba,\bg)
=[E\{m(0,\X,\ba)\}, \dots, E\{m(K,\X,\ba)\}]\trans,
\ese
  where
\bse
E\{m(k,\X,\ba)\}=\int m(k,\x,\ba) f_\X(\x,\bzeta) d\mu(\x).
\ese
Thus, 
\bse
\frac{\partial\btheta(\bzeta,\bb,\ba,\bg)}{\partial\ba\trans}&=&
\left[
E \left\{\frac{\partial m(0,\X,\ba)}{\partial\ba}\right\},
 \dots,
E \left\{\frac{\partial
    m(K,\X,\ba)}{\partial\ba}\right\}\right]\trans\Big|_{\ba=\ba_0},\\
\frac{\partial\btheta(\bzeta,\bb,\ba,\bg)}{\partial\bzeta\trans}&=&
\left[
E \left\{m(0,\X)\bS_\bzeta\right\},
 \dots,
E \left\{
    m(K,\X)\bS_\bzeta\right\}\right]\trans\Big|_{\bzeta=\bzeta_0},
\ese
while 
$\partial\btheta(\bzeta,\bb,\ba,\bg)/\partial\bb\trans=\0$ and
$\partial\btheta(\bzeta,\bb,\ba,\bg)/\partial\bg\trans=\0$.

Now consider 
\bse
\bphi=
\left[I(A=0)\frac{Y-m(0,\X)}{\pi(0,\X)}+m(0,\X_i), \dots,
 I(A=K)\frac{Y-m(K,\X)}{\pi(K,\X)}+m(K,\X_i)\right]\trans.
\ese
Denote $\wt\phi_k=I(A=k)\frac{Y-m(k,\X)}{\pi(k,\X)}+m(k,\X_i)$.
We can easily verify that
\bse
&&E(\phi_k\bS_\bb)\\
&=&E\left[\left\{I(A=k)\frac{Y-m(k,\X)}{\pi(k,\X)}+m(k,\X_i)\right\}
\left\{\sum_{l=0}^KI(A=l)\frac{\partial\pi(l,\x,\bb)/\partial\bb}{\pi(l,\x,\bb)}\right\}\right]\\
&=&E\left[\left\{I(A=k)\frac{Y-m(k,\X)}{\pi(k,\X)}\frac{\partial\pi(k,\x,\bb)/\partial\bb}{\pi(k,\x,\bb)}\right\}\right]
+E\left[m(k,\X_i)\left\{
\sum_{l=0}^KI(A=l)\frac{\partial\pi(l,\x,\bb)/\partial\bb}{\pi(l,\x,\bb)}\right\}\right]\\
&=&E\left[\{Y^k-m(k,\X)\}\frac{\partial\pi(k,\x,\bb)/\partial\bb}{\pi(k,\x,\bb)}\right]
+E\left[m(k,\X_i)
\left\{\frac{\partial
    \sum_{l=0}^K\pi(l,\x,\bb)}{\partial\bb}\right\}\right]\\
&=&\0,
\ese
and
\bse
&&E(\phi_k\bS_\bg)\\
&=&E\left(\left\{I(A=k)\frac{Y-m(k,\X)}{\pi(k,\X)}+m(k,\X)\right\}\right.\\
&&\left.\times\left[\sum_{l=0}^KI(A=l)\frac{\partial f_{\epsilon\mid(A,\X)}\{Y-m(l,\X,\ba),l,\X,\bg\}
/\partial\bg}{f_{\epsilon\mid(A,\X)}\{Y-m(l,\X,\ba),l,\X,\bg\}}\right]
\right)\\
&=&E\left(I(A=k)\frac{Y-m(k,\X)}{\pi(k,\X)}
\left[\sum_{l=0}^KI(A=l)\frac{\partial f_{\epsilon\mid(A,\X)}\{Y-m(l,\X,\ba),l,\X,\bg\}
/\partial\bg}{f_{\epsilon\mid(A,\X)}\{Y-m(l,\X,\ba),l,\X,\bg\}}\right]
\right)\\
&&+E\left(m(k,\X)\left[\sum_{l=0}^KI(A=l)\frac{\partial f_{\epsilon\mid(A,\X)}\{Y-m(l,\X,\ba),l,\X,\bg\}
/\partial\bg}{f_{\epsilon\mid(A,\X)}\{Y-m(l,\X,\ba),l,\X,\bg\}}\right]
\right)\\
&=&E\left\{\frac{\partial}{\partial\bg}\int
\epsilon f_{\epsilon\mid(A,\X)}(\epsilon,k,\X,\bg) d\epsilon
\right\}
+E\left[m(k,\X)\left\{\sum_{l=0}^K\pi(l,\X)
\frac{\partial}{\partial\bg}\int f_{\epsilon\mid(A,\X)}(\epsilon,l,\X,\bg)d\epsilon
\right\}
\right]\\
&=&\0.
\ese
Hence
$E(\bphi\bS_\bb\trans)=\0$ and $E(\bphi\bS_\bg\trans)=\0$. 
Further,
\bse
E(\phi_k\bS_\bzeta)
&=&E\left[
\left\{I(A=k)\frac{Y-m(k,\X)}{\pi(k,\X)}+m(k,\X_i)\right\}
\frac{\partial
  f_\X(\x,\bzeta)/\partial\bzeta}{f_\X(\x,\bzeta)}\right]\\
&=&\0+E\left\{m(k,\X_i)
\frac{\partial
  f_\X(\x,\bzeta)/\partial\bzeta}{f_\X(\x,\bzeta)}\right\}\\
&=&E\{m(k,\X)\bS_\bzeta(\X,\bzeta)\},
\ese
and
\bse
&&E(\phi_k\bS_\ba)\\
&=&E\left(
\left\{I(A=k)\frac{Y-m(k,\X)}{\pi(k,\X)}+m(k,\X)\right\}\right.\\
&&\left.\times
\left[\sum_{l=0}^KI(A=l)\frac{-\partial m(l,\X,\ba)}{\partial\ba}\frac{\partial
f_{\epsilon\mid(A,\X)}\{Y-m(l,\X,\ba),l,\X,\bg\}
/\partial\{Y-m(l,\X,\ba)\}}{f_{\epsilon\mid(A,\X)}\{Y-m(l,\X,\ba),l,\X,\bg\}}\right]\right)\\
&=&E\left[
\{Y^k-m(k,\X)\}
\frac{-\partial m(k,\X,\ba)}{\partial\ba}\frac{\partial
f_{\epsilon\mid(A,\X)}\{Y^k-m(k,\X,\ba),k,\X,\bg\}
/\partial\{Y^k-m(k,\X,\ba)\}}{f_{\epsilon\mid(A,\X)}\{Y^k-m(k,\X,\ba),k,\X,\bg\}}\right]\\
&&+E\left(
m(k,\X)
\left[\sum_{l=0}^K\pi(l,\X)\frac{-\partial m(l,\X,\ba)}{\partial\ba}\frac{\partial
f_{\epsilon\mid(A,\X)}\{Y^l-m(l,\X,\ba),l,\X,\bg\}
/\partial\{Y^l-m(l,\X,\ba)\}}{f_{\epsilon\mid(A,\X)}\{Y-m(l,\X,\ba),l,\X,\bg\}}\right]\right)\\
&=&E\left\{
\frac{-\partial m(k,\X,\ba)}{\partial\ba}\int 
\epsilon
\frac{\partial
f_{\epsilon\mid(A,\X)}(\epsilon,k,\X,\bg)}{\partial\epsilon} d\epsilon
\right\}\\
&&+E\left[
m(k,\X)
\left\{\sum_{l=0}^K\pi(l,\X)\frac{-\partial m(l,\X,\ba)}{\partial\ba}
\int
\frac{\partial
f_{\epsilon\mid(A,\X)}(\epsilon,l,\X,\bg)
}{\partial\epsilon}
d\epsilon
\right\}\right]\\
&=&E\left\{
\frac{\partial m(k,\X,\ba)}{\partial\ba}\right\},
\ese
where $\ba, \bzeta$ are evaluated at the true value $\ba_0, \bzeta_0$.
Therefore,
\bse
E(\bphi\bS_\bzeta\trans)= \left[
E\{m(0,\X)\bS_\bzeta(\X,\bzeta)\}, \dots, E\{m(K,\X)\bS_\bzeta(\X,\bzeta)\}\right]\trans
={\partial\btheta(\bzeta,\bb,\ba,\bg)}/{\partial\bzeta\trans}.
\ese
and
\bse
E(\bphi\bS_\ba\trans)= \left[E\{\partial
m(0,\x,\ba)/\partial\ba\}, \dots, \partial
E\{m(K,\x,\ba)/\partial\ba\}\right]\trans={\partial\btheta(\bzeta,\bb,\ba,\bg)}/{\partial\ba\trans}.
\ese 
Thus, $\bphi$ satisfies
$E(\bphi\bS_\bd\trans)=\partial\btheta(\bd)/\partial\bd\trans$. Because
the submodel is arbitrary,
$\bphi$ is an influence function of $\btheta$. 
We now try to obtain $\Pi(\bphi\mid\calT)$ so we can obtain the
efficient influence function. 
Further, we decompose $\bphi$ as 
$\bphi=(\bphi_1+\bphi_2+\bphi_3+\c)$, where
\bse
\bphi_1&=&\left(\begin{array}{c}
\frac{I(A=0)}{\pi(0,\X)}\left[Y-m(0,\X)+v(0,\X)\frac {f_{\epsilon\mid(A,\X)}'\{Y-m(0,\X),0,\X\}}
{f_{\epsilon\mid(A,\X)}\{Y-m(0,\X),0,\X\}}\right]\\
\vdots\\
\frac{I(A=K)}{\pi(K,\X)}\left[Y-m(K,\X)+v(K,\X)\frac {f_{\epsilon\mid(A,\X)}'\{Y-m(K,\X),K,\X\}}
{f_{\epsilon\mid(A,\X)}\{Y-m(K,\X),K,\X\}}\right]
\end{array}\right),\\
\bphi_2&=&-\left[\begin{array}{c}
\frac{I(A=0)}{\pi(0,\X)}v(0,\X)\frac {f_{\epsilon\mid(A,\X)}'\{Y-m(0,\X),0,\X\}}
{f_{\epsilon\mid(A,\X)}\{Y-m(0,\X),0,\X\}}\\
\vdots\\
\frac{I(A=K)}{\pi(K,\X)}v(K,\X)\frac {f_{\epsilon\mid(A,\X)}'\{Y-m(K,\X),K,\X\}}
{f_{\epsilon\mid(A,\X)}\{Y-m(K,\X),K,\X\}}
\end{array}\right],\\
\bphi_3&=&\left[\begin{array}{c}
m(0,\X)-E\{m(0,\X)\}\\
\vdots\\
m(K,\X)-E\{m(K,\X)\}\end{array}\right]
\ese
and $\c=[E\{m(0,\X)\}, \dots, E\{m(K,\X)\}]\trans$, where $v(k,\X)\equiv\var(Y^k\mid\X,A=k)$.
We can verify that $\bphi_1\in\calT_\bg$,
$\bphi_2\in\calT_\ba$, and $\bphi_3\in\calT_\bzeta$,
while $\c$ is a constant. Then
$\bphi-\c$ is the efficient influence function. 
Thus, the efficient variance is
$
\bSig_{\rm eff}=\var(\bphi),
$
where the $(k,l)$ entry of $\bSig_{\rm eff}$ is
\bse
\bSig_{\rm eff,k,l}=I(k=l)E\{v(k,\X)/\pi(k,\X)\}
+E([m(k,\X)-E\{m(k,\X)\}][m(l,\X)-E\{m(k,\X)\}]).
\ese
When $K=1$,
this agrees with the special case corresponding to the 
  binary treatments \citep{hahn:98}, and when $K>1$, with earlier
  results \citep{CATTANEO2010}. 

\subsection{Continuous treatment: derivations}\label{app:cont}

We prove all results under a general weight function
$w(A_j)$, where $w(A_j)=\sumi K_l(A_i-A_j)$ in the main paper.

\subsubsection{Convergence rate of $\wh\bb$}

\begin{proof}[Proof of Lemma \ref{lem:betaconv}]
	 From (\ref{eq:obj}), $\bb^*$ satisfies
\bse
\0
&=&E_{j}\left(
\left[E_i\left\{
\frac{K_l(A_i-A_j) \pi'_\bb(A_j,\X_i,\bb^*)}{\pi^2(A_j,\X_i,\bb^*)}
\B\trans(A_j,\X_i)\right\}
\right]\right.\\
&&\left.\times w(A_j)
E_i\left[
\left\{\frac{K_l(A_i-A_j)}{\pi(A_j,\X_i,\bb^*)}-1\right\}\B(A_j,\X_i)\right]
\right)\\
&=&E_{j}\left(
\left[E_i\left\{
\frac{\pi_0(A_j, \X_i)\pi'_\bb(A_j,\X_i,\bb^*)}{\pi^2(A_j,\X_i,\bb^*)}
\B\trans(A_j,\X_i)\right\}
\right]\right.\\
&&\left.\times w(A_j)
E_i\left[
\left\{\frac{\pi_0(A_j,\X_i)}{\pi(A_j,\X_i,\bb^*)}-1\right\}\B(A_j,\X_i)\right]
\right)+O(l^2),\\
&=&E_{j}\left(
\U(A_j,\bb^*)w(A_j)\right.\\
&&\left.\times 
E_i\left[
\left\{\frac{\pi_0(A_j,\X_i)}{\pi(A_j,\X_i,\bb^*)}-1\right\}\B(A_j,\X_i)\right]
\right)+O(l^2)\\
&=&E_{j}\left(
\U(A_j,\bb^*)w(A_j)
E_i\left[
\left\{\frac{\pi_0(A_j,\X_i)}{\pi(A_j,\X_i,\bb^*)}-1\right\}\B(A_j,\X_i)
\right]
\right)+O(l^2),
\ese
where
\bse
\U(a_j,\bb^*)\equiv
E\left\{
\frac{\pi_0(a_j,\X)\pi'_\bb(a_j,\X,\bb^*)}{\pi^2(a_j,\X,\bb^*)}
\B(a_j,\X)\trans\right\}.
\ese
We now investigate the convergence rate of $\wh\bb$
from
(\ref{eq:betaconint}). 
We note that
\bse
\0&=&\frac{1}{n}\sumj\\
&&\left[\frac{1}{n} \sumi\left\{
\frac{K_l(A_i-A_j) \pi'_\bb(A_j,\X_i,\wh\bb)}{\pi^2(A_j,\X_i,\wh\bb)}
\B\trans(A_j,\X_i)\right\}
\right]\\
&&\times w(A_j)
\left( \frac{1}{n^{1/2}}\sumi\left[
\left\{\frac{K_l(A_i-A_j)}{\pi(A_j,\X_i,\wh\bb)}-1\right\}\B(A_j,\X_i)\right]
\right)\\
&=&\frac{1}{n}\sumj\\
&&\left[E\left\{
\frac{K_l(A_i-a_j) \pi'_\bb(a_j,\X_i,\bb^*)}{\pi^2(a_j,\X_i,\bb^*)}
\B(a_j,\X_i)\trans
\right\}+o_p(1)
\right]\\
&&\times w(A_j)
\left( \frac{1}{n^{1/2}}\sumi\left[
\left\{\frac{K_l(A_i-A_j)}{\pi(A_j,\X_i,\bb^*)}-1\right\}\B(A_j,\X_i)\right]
\right)\\
&&-\frac{1}{n}\sumj\\
&&\left(
\left[E\left\{
\frac{K_l(A_i-a_j) \pi'_\bb(a_j,\X_i,\bb^*)}{\pi^2(a_j,\X_i,\bb^*)}
\B(a_j,\X_i)\trans\right\}
\right]^{\otimes2}\right.\\
&&\left. w(A_j)+o_p(1)\right)n^{1/2}(\wh\bb-\bb^*)\\
&=&\frac{1}{n^{3/2}}\sumi\sumj
\U(A_j,\bb^*)
w(A_j)\\
&&\times
\left\{\frac{K_l(A_i-A_j)}{\pi(A_j,\X_i,\bb^*)}-1\right\}\B(A_j,\X_i)
\\
&&-E
\left[\{\U(A_j,\bb^*)\}^{\otimes2}w(A_j)\right]n^{1/2}(\wh\bb-\bb^*)+o_p(1).
\ese
We have
\bse
&&\frac{1}{n^{3/2}}\sumi\sumj
\U(A_j,\bb^*)w(A_j)
\left\{\frac{K_l(A_i-A_j)}{\pi(A_j,\X_i,\bb^*)}-1\right\}\B(A_j,\X_i)\\
&=&\frac{1}{n^{1/2}}\sumj
\U(a_j,\bb^*)w(a_j)
E_i\left[\left\{\frac{K_l(A_i-a_j)}{\pi(a_j,\X_i,\bb^*)}-1\right\}\B(a_j,\X_i)\right]\\
&&+\frac{1}{n^{1/2}}\sumi
E_j\left[\U(A_j,\bb^*)w(A_j)
\left\{\frac{K_l(a_i-A_j)}{\pi(A_j,\x_i,\bb^*)}-1\right\}\B(A_j,\x_i)\right]\\
&&-n^{1/2}
E_{ij}\left[\U(A_j,\bb^*)w(A_j)
\left\{\frac{K_l(A_i-A_j)}{\pi(A_j,\X_i,\bb^*)}-1\right\}\B(A_j,\X_i)\right]
+o_p(1)\\
&=&\frac{1}{n^{1/2}}\sumj
\U(a_j,\bb^*)w(a_j)
E_i\left\{\frac{\pi_0(a_j,\X_i)}{\pi(a_j,\X_i,\bb^*)}\B(a_j,\X_i)-\B(a_j,\X_i)\right\}\\
&&+\frac{1}{n^{1/2}}\sumi
\left[f_A(a_i)\U(a_i,\bb^*)w(a_i)
\frac{\B(a_i,\x_i)}{\pi(a_i,\x_i,\bb^*)}
-E_j\left\{\U(A_j,\bb^*)w(A_j)\B(A_j,\x_i)
\right\}\right]
\\
&&-n^{1/2}
E_{j}\left[\U(A_j,\bb^*)w(A_j)
E_i\left\{\frac{\pi_0(A_j,\X_i)}{\pi(A_j,\X_i,\bb^*)}\B(A_j,\X_i)-\B(A_j,\X_i)\right\}
\right]\\
&&+o_p(1)+O_p(n^{1/2}l^2).
\ese

Thus, when $nl^4\to0$, we get
\bse
&&E\left\{\U(A_j,\bb^*)^{\otimes2}w(A_j)\right\}n^{1/2}(\wh\bb-\bb^*)\\
&=&\frac{1}{n^{1/2}}\sumj\U(a_j,\bb^*)
w(a_j)
E_i\left[\left\{\frac{\pi_0(a_j,\X_i)}{\pi(a_j,\X_i,\bb^*)}-1\right\}\B(a_j,\X_i)\right]\\
&&+\frac{1}{n^{1/2}}\sumi
\left[f_A(a_i)\U(a_i,\bb^*)w(a_i)
\frac{\B(a_i,\x_i)}{\pi(a_i,\x_i,\bb^*)}
-E_j\{\U(A_j,\bb^*)w(A_j)\B(A_j,\x_i)
\}\right]\\
&&-n^{1/2}
E_j\left(\U(A_j,\bb^*)w(A_j)
E_i\left[\left\{\frac{\pi_0(A_j,\X_i)}{\pi(A_j,\X_i,\bb^*)}-1\right\}\B(A_j,\X_i)\right]
\right)\\
&&+o_p(1)+O_p(n^{1/2}l^2)\\
&=&\frac{1}{n^{1/2}}\sumi
\U(a_i,\bb^*)w(a_i)
E_k\left[\left\{\frac{\pi_0(a_i,\X_k)}{\pi(a_i,\X_k,\bb^*)}-1\right\}\B(a_i,\X_k)\right]\\
&&+\frac{1}{n^{1/2}}\sumi
\left[f_A(a_i)\U(a_i,\bb^*)w(a_i)
\frac{\B(a_i,\x_i)}{\pi(a_i,\x_i,\bb^*)}
-E_{j}\left\{\U(A_j,\bb^*)w(A_j)\B(A_j,\x_i)
\right\}\right]\\
&&+o_p(1).
\ese
Obviously, 
\bse
E_i\left(\U(A_i,\bb^*)w(A_i)
E_k\left[\left\{\frac{\pi_0(A_i,\X_k)}{\pi(A_i,\X_k,\bb^*)}-1\right\}\B(A_i,\X_k)\right]\right)
=O(l^2)
\ese
due to the definition of $\bb^*$. Further, we can verify that
\bse
&&E_i\left[f_A(A_i)\U(A_i,\bb^*)w(A_i)
\frac{\B(A_i,\X_i)}{\pi(A_i,\X_i,\bb^*)}
-E_{j}\left\{\U(A_j,\bb^*)w(A_j)\B(A_j,\X_i)
\right\}\right]\\
&=&E_{i,j}\left\{\U(A_j,\bb^*)w(A_j)
\frac{\pi_0(A_j,\X_i) \B(A_j,\X_i)}{\pi(A_j,\X_i,\bb^*)}\right\}
-E_{i,j}\left\{\U(A_j,\bb^*)w(A_j)\B(A_j,\X_i)
\right\}\\
&=&
E_{j}\left(\U(A_j,\bb^*)w(A_j)
E_i\left[\left\{\frac{\pi_0(A_j,\X_i)}{\pi(A_j,\X_i,\bb^*)}-1\right\}\B(A_j,\X_i)\right]\right)\\
&=&O(l^2)
\ese
also due to the definition of $\bb^*$. Thus, as long as $nl^4\to0$, $\wh\bb-\bb^*=O_p(n^{-1/2})$. 
\end{proof}

\subsubsection{Robustness and asymptotic bias and variance}
\begin{proof}[Proof of Theorem \ref{th:a}]
When model (\ref{eq:contpispec}) holds, 
we can easily check that
the expectation of the left hand side of (\ref{eq:betacon1}) at the
true parameter value $\bb_0$ and any function $m(a,\x)=\B(a,\x)\trans\bg$ satisfies
\bse
&&E\left[\left\{\frac{K_h(A_i-a)}{\pi_0(a,\X_i)}-1\right\}m(a,\X_i)
\right]\\
&=&E\left[ \left\{\frac{E\{K_h(A_i-a)\mid\X_i\}}{\pi_0(a,\X_i)}-1\right\} m(a,\X_i)
\right]\\
&=&E\left[ \left\{\frac{\int K_h(A_i-a)\pi_0(A_i,\X_i)dA_i}{\pi_0(a,\X_i)}-1\right\} m(a,\X_i)
\right]\\
&=&E\left[ \left\{\frac{\int K(t)\pi_0(a+ht,\X_i)dt}{\pi_0(a,\X_i)}-1\right\} m(a,\X_i)
\right]\\
&=&E\left[ \left\{\frac{\int K(t)\pi_0(a,\X_i)dt}{\pi_0(a,\X_i)}-1\right\} m(a,\X_i)
\right]+O(h^2)\\
&=&O(h^2).
\ese
Thus, because the
nonparametric estimation convergence rate is slower than
$O_p(n^{-1/2})$, by Lemma \ref{lem:betaconv} we can fix $\bb$ at
$\bb_0$ in the following analysis, and the first order bias and variance property of
$\wh\theta(a)$ will not be affected.

Hence, for (\ref{eq:thetacon}), we have
\bse
E\{\wh\theta(a)\}
&=&E\left\{
\frac{K_h(A_i-a)Y_i}{\pi_0(a,\X_i)}\right\}+O(n^{-1/2})\\
&=&E\left\{
\frac{K_h(A_i-a)Y_i(A_i)}{\pi_0(a,\X_i)}\right\}+O(n^{-1/2})\\
&=&E\left\{
\frac{K_h(A_i-a)m( A_i, \X_i)}{\pi_0(a,\X_i)}\right\}+O(n^{-1/2})\\
&=&E\left[m(a,\X_i)  
+\frac{\partial^2 \{\pi_0(a,\X_i) m(a,\X_i)\}}{\pi_0(a,\X_i)\partial
  a^2}
\frac{h^2}{2}\int
t^2K(t)dt\right]\\
&&+O(h^4+n^{-1/2})\\
&=&\theta(a)
+E\left[\frac{\partial^2 \{\pi_0(a,\X_i) m(a,\X_i)\}}{\pi_0(a,\X_i)\partial
  a^2}\right]
\frac{h^2}{2}\int
t^2K(t)dt+O(h^4+n^{-1/2}).
\ese

The variance is calculated as
\bse
\var\{\wh\theta(a)\}
=\var\left[
n^{-1}\sumi
\left\{
\frac{K_h(A_i-a)Y_i}{\pi_0(a,\X_i)}
\right\}+O_p(n^{-1/2})
\right].
\ese
Now, recall that the variance of $Y_i(A_i)$ conditional on $\X_i, A_i$
is denoted $\sigma^2(A_i, \X_i)$, then
\bse
&&E\left[\left\{
\frac{K_h(A_i-a)Y_i}{\pi_0(a,\X_i)}
\right\}^2\right]\\
&=&E\left[\left\{
\frac{K_h(A_i-a)}{\pi_0(a,\X_i)}
\right\}^2
\{m^2(A_i,\X_i)+\sigma^2(A_i,\X_i)\}\right]\\
&=&\frac{\int K^2(t)dt}{h}E\left\{
\frac{m^2(a,\X_i)+\sigma^2(a,\X_i)}{\pi_0(a,\X_i)}
\right\}
+O(h).
\ese
Thus,
\bse
\var\{\wh\theta(a)\}
&=&\var\left[
n^{-1}\sumi
\left\{
\frac{K_h(A_i-a)Y_i}{\pi_0(a,\X_i)}
\right\}+O_p(n^{-1/2})
\right]\\
&=&\frac{\int K^2(t)dt}{nh}E\left\{
\frac{m^2(a,\X_i)+\sigma^2(a,\X_i)}{\pi_0(a,\X_i)}
\right\}\\
&&+O(n^{-1}h+n^{-1}+n^{-1}h^{-1/2}).
\ese
The asymptotic normality is shown in Section \ref{app:contas} below.
\end{proof}

\begin{proof}[Proof of Theorem \ref{th:biasvar1}]
	When model (\ref{eq:contmspec}) is correct, then $\wh\bb$ converges to
$\bb^*$ at root-$n$ rate (Lemma \ref{lem:betaconv}). Thus, 
\bse
E\{\wh\theta(a)\}
&=&E\left\{
\frac{K_h(A_i-a)Y_i}{\pi(a,\X_i,\bb^*)}
\right\}+O(n^{-1/2})\n\\
&=&E\left\{
\frac{K_h(A_i-a)Y_i(A_i)}{\pi(a,\X_i,\bb^*)}
\right\}+O(n^{-1/2})\n\\
&=&E\left[
\frac{K_h(A_i-a)m(A_i, \X_i)}{\pi(a,\X_i,\bb^*)}
\right]+O(n^{-1/2})\n\\
&=&E\left[
\frac{\pi_0(a,\X_i)m(a, \X_i)}{\pi(a,\X_i,\bb^*)}
\right]
+\frac{\int t^2K(t)dt}{2}h^2\n\\
&&\times E
\left[ \frac{\partial^2 \{m(a,\X_i)\pi_0(a,\X_i)\}}{\pi(a,\X_i,\bb^*)\partial a^2
}\right]+O(n^{-1/2})\n\\
&=&E\left[\left\{\frac{\pi_0(a,\X_i)}{\pi(a,\X_i,\bb^*)}-1\right\}m(a, \X_i)
\right]
+E\{m(a,\X_i)\}\n\\
&&+\frac{\int t^2K(t)dt}{2}h^2
 E
\left[\frac{\partial^2 \{m(a,\X_i)\pi_0(a,\X_i)\}}{\pi(a,\X_i,\bb^*) \partial a^2
}\right]+O(n^{-1/2})\n\\
&=&E\{m(a,\X_i)\}
+\frac{\int t^2K(t)dt}{2}h^2
 E
\left[\frac{\partial^2 \{m(a,\X_i)\pi_0(a,\X_i)\}}{\pi(a,\X_i,\bb^*) \partial a^2
}\right]+O(n^{-1/2}).
\ese
The variance is calculated as
\bse
\var\{\wh\theta(a)\}
=\var\left[
n^{-1}\sumi
\frac{K_h(A_i-a)Y_i}{\pi(a,\X_i,\bb^*)}
+O_p(n^{-1/2})
\right].
\ese
Then
\bse
&&E\left[\left\{
\frac{K_h(A_i-a)Y_i}{\pi(a,\X_i,\bb^*)}
\right\}^2\right]\\
&=&E\left[\left\{
\frac{K_h(A_i-a)}{\pi(a,\X_i,\bb^*)}
\right\}^2
\{m^2(A_i,\X_i)+\sigma^2(A_i,\X_i)\}\right]\\
&=&\frac{\int K^2(t)dt}{h}E\left[
\frac{\pi_0(a,\X_i)\{m^2(a,\X_i)+\sigma^2(a,\X_i)\}}{\pi^2(a,\X_i,\bb^*)}
\right]
+O(h).
\ese
Thus,
\bse
\var\{\wh\theta(a)\}
&=&\var\left\{
n^{-1}\sumi
\frac{K_h(A_i-a)Y_i}{\pi(a,\X_i,\bb^*)}
+O_p(n^{-1/2})
\right\}\\
&=&\frac{\int K^2(t)dt}{nh}E\left[
\frac{\pi_0(a,\X_i)\{m^2(a,\X_i)+\sigma^2(a,\X_i)\}}{\pi^2(a,\X_i,\bb^*)}
\right]\\
&&+O(n^{-1}h+n^{-1}+n^{-1}h^{-1/2}).
\ese
The asymptotic normality is shown in Section \ref{app:contas} below.
\end{proof}

\begin{proof}[Proof of Theorem \ref{th:biasvar2}]
\bse
&&\cov\left\{n^{-1}\sumi\frac{K_h(A_i-a)Y_i}{\pi(a,\X_i,\wh\bb)},
n^{-1}\sumi\frac{K_h(A_i-b)Y_i}{\pi(b,\X_i,\wh\bb)}\right\}\\
&=&E\left[\left\{n^{-1}\sumi\frac{K_h(A_i-a)Y_i}{\pi(a,\X_i,\wh\bb)}\right\}
\left\{n^{-1}\sumi\frac{K_h(A_i-b)Y_i}{\pi(b,\X_i,\wh\bb)}\right\}\right]\\
&&-E\left\{n^{-1}\sumi\frac{K_h(A_i-a)Y_i}{\pi(a,\X_i,\wh\bb)}\right\}
E\left\{n^{-1}\sumi\frac{K_h(A_i-b)Y_i}{\pi(b,\X_i,\wh\bb)}\right\}\\
&=&n^{-2}\sumi E\left\{\frac{K_h(A_i-a)K_h(A_i-b)Y_i^2
}{\pi(a,\X_i,\wh\bb)\pi(b,\X_i,\wh\bb)}\right\}+
n^{-2}\sum_{i\ne j, i.j=1}^n E\left\{\frac{K_h(A_i-a)Y_i}{\pi(a,\X_i,\wh\bb)}
\frac{K_h(A_j-b)Y_j}{\pi(b,\X_j,\wh\bb)}\right\}\\
&&-E\left\{\frac{K_h(A_i-a)Y_i}{\pi(a,\X_i,\wh\bb)}\right\}
E\left\{\frac{K_h(A_i-b)Y_i}{\pi(b,\X_i,\wh\bb)}\right\}\\
&=&n^{-1} E\left\{\frac{K_h(A_i-a)K_h(A_i-b)Y_i^2
}{\pi(a,\X_i,\wh\bb)\pi(b,\X_i,\wh\bb)}\right\}
-n^{-1} E\left\{\frac{K_h(A_i-a)Y_i}{\pi(a,\X_i,\wh\bb)}\right\}
E\left\{\frac{K_h(A_i-b)Y_i}{\pi(b,\X_i,\wh\bb)}\right\}\\
&=&n^{-1} E\left\{\frac{K_h(A_i-a)K_h(A_i-b)Y_i^2
}{\pi(a,\X_i,\wh\bb)\pi(b,\X_i,\wh\bb)}\right\}
-n^{-1}E\{\wh\theta(a)\}E\{\wh\theta(b)\}\\
&=&n^{-1} E\left\{\frac{K_h(A_i-a)K_h(A_i-b)Y_i^2
}{\pi(a,\X_i,\wh\bb)\pi(b,\X_i,\wh\bb)}\right\}
-n^{-1}\{\theta(a)\theta(b)+O(h^2)\}.
\ese
When $a$ and $b$ are sufficiently close, so that
$c\equiv(a-b)/h\in(-2,1)$, we have
\bse
&&E\left\{\frac{K_h(A_i-a)K_h(A_i-b)Y_i^2}{\pi(a,\X_i,\wh\bb)\pi(b,\X_i,\wh\bb)}\right\}\n\\
&=&E\left\{\frac{K_h(A_i-a)K_h(A_i-b)
\{m^2(A_i,\X_i)+\sigma^2(A_i,\X_i)\}
}{\pi(a,\X_i,\wh\bb)\pi(b,\X_i,\wh\bb)}\right\}\\
&=&h^{-1}E\int_0^1\frac{K(t)K(t+c)\{m^2(a+ht,\X_i)+\sigma^2(a+ht,\X_i)\}}
{\pi(a,\X_i,\wh\bb)\pi(b,\X_i,\wh\bb)}\pi_0(a+ht,\X_i)dt\\
&=&h^{-1}E\int_0^1
\frac{K(t)K(t+c)\{m^2(a,\X_i)+\sigma^2(a,\X_i)\}}
{\pi(a,\X_i,\wh\bb)\pi(b,\X_i,\wh\bb)}\pi_0(a,\X_i)dt\\
&&+E\int_0^1K(t)K(t+c)
\left\{2m(a,\X_i)m'_a(a,\X_i) \pi_0(a,\X_i)+m^2(a,\X_i) \pi_{0a}'(a,\X_i)\right.\\
&&\left.+2\sigma(a,\X_i)\sigma'_a(a,\X_i) \pi_0(a,\X_i)+\sigma^2(a,\X_i) \pi_{0a}'(a,\X_i)\right\}t/\{\pi(a,\X_i,\wh\bb)\pi(b,\X_i,\wh\bb)\}
dt+O(h)\\
&=&h^{-1}E\int_0^1
\frac{K(t)K(t+c)\{m^2(a,\X_i)+\sigma^2(a,\X_i)\}}
{\pi(a,\X_i,\bb^*)\pi(b,\X_i,\bb^*)}\pi_0(a,\X_i)dt\\
&&+E\int_0^1K(t)K(t+c)
\left\{2m(a,\X_i)m'_a(a,\X_i) \pi_0(a,\X_i)+m^2(a,\X_i) \pi_{0a}'(a,\X_i)\right.\\
&&\left.+2\sigma(a,\X_i)\sigma'_a(a,\X_i) \pi_0(a,\X_i)+\sigma^2(a,\X_i) \pi_{0a}'(a,\X_i)\right\}t/\{\pi(a,\X_i,\bb^*)\pi(b,\X_i,\bb^*)\}
dt\\
&&+O(h+h^{-1}n^{-1/2}).
\ese
Note that when $c\notin(-2,1)$, $K(t)K(t+c)=0$ for all $t\notin[-1,1]$ hence
the above expression still holds. 
Thus, we obtain
\bse
&&\cov\{\wh\theta(a),\wh\theta(b)\}\\
&=&
(nh)^{-1}E\int_0^1
\frac{K(t)K(t+c)\{m^2(a,\X_i)+\sigma^2(a,\X_i)\}}
{\pi(a,\X_i,\bb^*)\pi(b,\X_i,\bb^*)}\pi_0(a,\X_i)dt\\
&&+n^{-1}E\int_0^1K(t)K(t+c)
\left\{2m(a,\X_i)m'_a(a,\X_i) \pi_0(a,\X_i)+m^2(a,\X_i) \pi_{0a}'(a,\X_i)\right.\\
&&\left.+2\sigma(a,\X_i)\sigma'_a(a,\X_i) \pi_0(a,\X_i)+\sigma^2(a,\X_i) \pi_{0a}'(a,\X_i)\right\}t/\{\pi(a,\X_i,\bb^*)\pi(b,\X_i,\bb^*)\}
dt\\
&&-n^{-1}\theta(a)\theta(b)
+O(n^{-1}h+h^{-1}n^{-3/2}).
\ese

The asymptotic normality is shown in Section \ref{app:contas} below.
\end{proof}

\subsubsection{Asymptotic distribution of $\wh\theta(a)$}\label{app:contas}
\begin{proof}[Proof of asymptotic normality, Theorems \ref{th:a}-\ref{th:biasvar2}]
	When (\ref{eq:contpispec}) is correct, define
\bse
{\rm bias}\{\wh\theta(a)\}=\frac{h^2}{2}E\left[
\frac{\partial^2 \{\pi_0(a,\X_i) m(a,\X_i)\}}{\pi_0(a,\X_i)\partial
  a^2}\right]
\int
t^2K(t)dt.
\ese
On the other hand, when (\ref{eq:contmspec}) is correct, define
\bse
{\rm bias}\{\wh\theta(a)\}&=&\frac{h^2}{2}E
\left[\frac{\partial^2 \{\pi_0(a,\X_i)m(a,\X_i)\}}{\pi(a,\X_i,\bb^*) \partial a^2
}\right]\int t^2K(t)dt.
\ese
Regardless (\ref{eq:contpispec})  or (\ref{eq:contmspec}) is correct,
define
\bse
\var(\wh\theta)
=\frac{\int K^2(t)dt}{nh}E\left[
\frac{\pi_0(a,\X_i)\{m^2(a,\X_i)+\sigma^2(a,\X_i)\}}{\pi^2(a,\X_i,\bb^*)}
\right].
\ese
Note that when (\ref{eq:contpispec}) is correct, it degenerates to
\bse
\var\{\wh\theta(a)\}=\frac{\int K^2(t)dt}{nh}E\left\{
\frac{m^2(a,\X_i)+\sigma^2(a,\X_i)}{\pi_0(a,\X_i)}
\right\}.
\ese
Then
\bse
&&\left[\wh\theta(a)-\theta(a)-{\rm bias}\{\wh\theta(a)\}\right]\\
&=&n^{-1}\sumi \frac{K_h(A_i-a)Y_i}{\pi(a,\X_i,\bb^*)}
-\theta(a)
-{\rm bias}\{\wh\theta(a)\} +O_p(n^{-1/2})\\
&=&n^{-1}\sumi \left[\frac{K_h(A_i-a)Y_i}{\pi(a,\X_i,\bb^*)}
-E\left\{\frac{K_h(A_i-a)Y_i}{\pi(a,\X_i,\bb^*)}\right\}\right]
+O_p(h^4+n^{-1/2}).
\ese
Thus, when $n\to\infty$, following the variance result, we get that
\bse
\sqrt{nh}\left[\wh\theta(a)-\theta(a)-{\rm bias}\{\wh\theta(a)\} 
\right]
\ese
converges to a normal distribution with mean zero and variance
$nh\var\{\wh\theta(a)\}$.

Consider an arbitrary linear combination
$\sum_{j=1}^Jc_j\wh\theta(a_j)$.
Then
\bse
&&\left[\sum_{j=1}^Jc_j\wh\theta(a_j)-\sum_{j=1}^Jc_j\theta(a_j)-{\rm bias}\left\{\sum_{j=1}^Jc_j\wh\theta(a_j)\right\}\right]\\
&=&n^{-1}\sumi \sum_{j=1}^Jc_j\frac{K_h(A_i-a_j)Y_i}{\pi(a_j,\X_i,\bb^*)}
-\sum_{j=1}^Jc_j\theta(a_j)
-\sum_{j=1}^Jc_j{\rm bias}\{\wh\theta(a)\} +O_p(n^{-1/2})\\
&=&n^{-1}\sumi \sum_{j=1}^Jc_j\left[\frac{K_h(A_i-a_j)Y_i}{\pi(a_j,\X_i,\bb^*)}
-E\left\{\frac{K_h(A_i-a_j)Y_i}{\pi(a_j,\X_i,\bb^*)}\right\}\right]
+O_p(h^4+n^{-1/2})
\ese
converges to a normal distribution with mean zero. To compute its
variance, we compute $\cov\{\wh\theta(a),\wh\theta(b)\}$ for arbitrary
$a, b$ below.

Let $\cov\{\wh\theta(a),\wh\theta(b)\}$ be given as the leading term in 
(\ref{eq:cov}). Note that when (\ref{eq:contpispec}) is correct, it
degenerates to
\bse
&&\cov\{\wh\theta(a),\wh\theta(b)\}\\
&=&
(nh)^{-1}E\int_0^1
\frac{K(t)K(t+c)\{m^2(a,\X_i)+\sigma^2(a,\X_i)\}}
{\pi_0(b,\X_i)}dt\\
&&+n^{-1}E\int_0^1K(t)K(t+c)
\left\{2m(a,\X_i)m'_a(a,\X_i) \pi_0(a,\X_i)+m^2(a,\X_i) \pi_{0a}'(a,\X_i)\right.\\
&&\left.+2\sigma(a,\X_i)\sigma'_a(a,\X_i) \pi_0(a,\X_i)+\sigma^2(a,\X_i) \pi_{0a}'(a,\X_i)\right\}t/\{\pi_0(a,\X_i)\pi_0(b,\X_i)\}
dt\\
&&-n^{-1}\theta(a)\theta(b).
\ese
Here $c=(a-b)/h$. Then the above analysis leads to that
$
\wh\theta(a)-\theta(a)
$
is asymptotically a Gaussian process with mean 
given by ${\rm bias}\{\wh\theta(a)\}$ and variance-covariance function
given in $\cov\{\wh\theta(a),\wh\theta(b)\}$.
\end{proof}

\subsubsection{Variance estimation}\label{app:varest}
By Theorem \ref{th:biasvar1},
\bse
\var\{\wh\theta(a)\}
&=& 
(nh)^{-1}E\int_0^1
\frac{K(t)^2\{m^2(a,\X_i)+\sigma^2(a,\X_i)\}}
{\pi(a,\X_i,\bb^*)^2}\pi_0(a,\X_i)dt
+ O(n^{-1} + n^{-1}h + n^{-3/2}h^{-1}) \\
&=&
\frac{\int_0^1 K(t)^2 dt}{nh}
E\Bigg[
\frac{\pi_0(a,\X_i)\{m^2(a,\X_i)+\sigma^2(a,\X_i)\}}
{\pi(a,\X_i,\bb^*)^2} \Bigg]
+ O(n^{-1} + n^{-1}h + n^{-3/2}h^{-1}).
\ese
Thus, an estimator of this variance is obtained as 
\bse
\wh\var\{\wh\theta(a)\}
&=& \frac{\int_0^1 K(t)^2 dt}{nh} n^{-1}\sumi\Bigg[
\frac{K_h(A_i-a)
\left\{\wh m^2(A_i,\X_i)+\{Y_i-\wh m(A_i,\X_i)\}^2\right\}}
{\pi(a,\X_i,\wh \bb)^2} \Bigg].
\ese
Let $\wh m(A_i,\X_i) = Y_i$. Then the above estimator becomes
\bse
\wh\var\{\wh\theta(a)\}
&=& 
\frac{\int_0^1 K(t)^2 dt}{nh} n^{-1}\sumi
\frac{K_h(A_i-a)Y_i^2}{\pi(a,\X_i,\wh \bb)^2}.
\ese
Its expectation is
\bse
&&
\frac{\int_0^1 K(t)^2 dt}{nh} E\left\{n^{-1}\sumi
\frac{K_h(A_i-a)Y_i^2}{\pi(a,\X_i,\wh \bb)^2} \right\} \\
&=&
\frac{\int_0^1 K(t)^2 dt}{nh} E\left\{n^{-1}\sumi
\frac{K_h(A_i-a)Y_i^2}{\pi(a,\X_i,\bb^*)^2}
+ O(n^{-1/2}) \right\} \\
&=&
\frac{\int_0^1 K(t)^2 dt}{nh} E\left\{
\frac{K_h(A_i-a)Y_i^2}{\pi(a,\X_i,\bb^*)^2} \right\}
+ O(n^{-3/2}h^{-1})\\
&=&
\frac{\int_0^1 K(t)^2 dt}{nh} E\Bigg[
\frac{K_h(A_i-a)
\left\{m^2(A_i,\X_i) + \sigma^2(A_i,\X_i)\right\}}
{\pi(a,\X_i,\bb^*)^2} \Bigg]
+ O(n^{-3/2}h^{-1})\\
&=&
\frac{\int_0^1 K(t)^2 dt}{nh} \Bigg(E\bigg[
\frac{\pi_0(a,\X_i)
\left\{m^2(a,\X_i) + \sigma^2(a,\X_i)\right\}}
{\pi(a,\X_i,\bb^*)^2} \bigg] \\
&&+ \frac{\int t^2K(t)dt}{2}h^2 \times E
\left[\frac{\partial^2}{\partial a^2} 
\frac{\pi_0(a,\X_i)
\left\{m^2(a,\X_i) + \sigma^2(a,\X_i)\right\}}
{\pi(a,\X_i,\bb^*)^2} \right] \Bigg)
+ O(n^{-3/2}h^{-1}) \\
&=&
\frac{\int_0^1 K(t)^2 dt}{nh} E\bigg[
\frac{\pi_0(a,\X_i)
\left\{m^2(a,\X_i) + \sigma^2(a,\X_i)\right\}}
{\pi(a,\X_i,\bb^*)^2} \bigg]
+ O(n^{-1}h+n^{-3/2}h^{-1}).
\ese
Following Remark \ref{rem:con}, an alternative variance estimator is:
\bse
\wh \var \{ \wh \theta(a) \} = \frac{\int_0^1 K(t)^2 dt}{nh}
\left\{ \sumi  \frac{K_h(A_i-a)}{\pi(a,\X_i,\wh \bb)} \right\}^{-1}
\sumi  \frac{K_h(A_i-a)Y_i^2}{\pi(a,\X_i,\wh \bb)^2} .
\ese
\newpage

\section{Simulation results: Tables}\label{app:tab}
\begin{table}[h]
\begin{center}
\caption{Results based on 1000 replicates for the estimation of
  contrasts $\theta_k-\theta_0$,
$k=1,2,3$ with balancing estimator
  proposed using model $\pi(\cdot)$ and basis of $m(\cdot)$, which are
  either correctly specified or misspecified. Last  blocks contain
  maximum likelihood based IPW (ML-IPW) and augmented IPW (DR) estimators. Sample size $n=500$.}\vskip 0.2cm 
\label{tab:cate1}
\begin{tabular}{cccccc}
\hline
$\theta_k-\theta_0$&bias&sd&MSE&$\wh{\rm sd}$&95\%\\ \hline
&\multicolumn{5}{c}{$m$, $\pi$ correct}\\ 
$k=1$&0.3160&2.6185&6.9566&2.6078&0.9520\\
$k=2$&0.3211&2.6183&6.9586&2.6073&0.9510\\
$k=3$&0.3167&2.6173&6.9503&2.6075&0.9520\\
&\multicolumn{5}{c}{$\pi$ correct}\\ 
$k=1$&1.3666&7.4357&57.1567&6.2238&0.9110\\
$k=2$&1.2198&7.1377&52.4345&5.7876&0.8940\\
$k=3$&1.3181&7.0207&51.0281&5.7158&0.9000\\
&\multicolumn{5}{c}{$m$ correct}\\ 
$k=1$&2.1145&3.4709&16.5182&3.5342&0.9550\\
$k=2$&2.1204&3.4748&16.5701&3.5341&0.9560\\
$k=3$&2.1154&3.4711&16.5235&3.5339&0.9530\\
&\multicolumn{5}{c}{$m$, $\pi$ misspecified}\\
$k=1$&3.2163&7.8904&72.6030&7.0024&0.9150\\
$k=2$&3.0839&7.6868&68.5982&6.5804&0.9020\\
$k=3$&3.1900&7.4916&66.3006&6.5237&0.9060\\
&\multicolumn{5}{c}{ML-IPW, $\pi$ correct}\\
$k=1$& 0.0842&16.5578&274.1668&16.3236&0.9650\\
$k=2$&0.4053&14.3483&206.0379&14.0882&0.9530\\
$k=3$&0.1948&14.0600&197.7213&14.0238&0.9520\\
&\multicolumn{5}{c}{DR, $m$, $\pi$ correct}\\
$k=1$&0.040&2.352&5.533&2.451&0.962\\
$k=2$&0.045&2.351&5.529&2.450&0.962\\
$k=3$&0.041&2.349&5.520&2.450&0.964\\
\hline
\end{tabular}
\end{center}
\end{table}

\begin{table}[h]
\begin{center}
\caption{Results based on 1000 replicates for the estimation of
  contrasts $\theta_k-\theta_0$,
$k=1,2,3$ with balancing estimator
  proposed using model $\pi(\cdot)$ and basis of $m(\cdot)$, which are
  either correctly specified or misspecified. Last  blocks contain
  maximum likelihood based IPW (ML-IPW) and augmented IPW (DR) estimators. Sample size $n=1000$.}\vskip 0.2cm 
\label{tab:cate2}
\begin{tabular}{cccccc}
\hline
$\theta_k-\theta_0$&bias&sd&MSE&$\wh{\rm sd}$&95\%\\ \hline
&\multicolumn{5}{c}{$m$, $\pi$ correct}\\ 
$k=1$&0.1233&1.9123&3.6720&1.8477&0.9380\\
$k=2$&0.1273&1.9111&3.6686&1.8472&0.9370\\
$k=3$&0.1233&1.9092&3.6604&1.8471&0.9380\\
&\multicolumn{5}{c}{$\pi$ correct}\\ 
$k=1$&0.3756&5.1489&26.6518&4.4066&0.9160\\
$k=2$&0.4287&4.7061&22.3316&4.0946&0.9070\\
$k=3$&0.3302&4.7935&23.0868&4.0950&0.9110\\
&\multicolumn{5}{c}{$m$ correct}\\ 
$k=1$&1.2285&2.2205&6.4397&2.2226&0.9360\\
$k=2$&1.2325&2.2225&6.4588&2.2222&0.9350\\
$k=3$&1.2284&2.2206&6.4400&2.2220&0.9360\\
&\multicolumn{5}{c}{$m$, $\pi$ misspecified}\\
$k=1$&1.4565&5.4090&31.3788&4.6882&0.9080\\
$k=2$&1.5062&4.9498&26.7694&4.3911&0.9050\\
$k=3$&1.4004&5.0466&27.4296&4.3925&0.9150\\
&\multicolumn{5}{c}{ML-IPW, $\pi$ correct}\\
$k=1$&0.0974&11.5132&132.5634&10.8010&0.9540\\
$k=2$&0.2635&10.2896&105.9450&9.4923&0.9510\\
$k=3$&0.0573&10.4489&109.1838&9.4719&0.9480\\
&\multicolumn{5}{c}{DR, $m$, $\pi$ correct}\\
$k=1$&0.048&1.747&3.056&1.737&0.947\\
$k=2$&0.052&1.747&3.054&1.736&0.947 \\
$k=3$&0.048&1.746&3.050&1.736&0.949\\ \hline
\end{tabular}
\end{center}
\end{table}

\begin{table}[h]
\begin{center}
\caption{Results based on 1000 replicates for the estimation of
  contrasts $\theta_k-\theta_0$,
$k=1,2,3$ with balancing estimator
  proposed using model $\pi(\cdot)$ and basis of $m(\cdot)$, which are
  either correctly specified or misspecified. Last  blocks contain
  maximum likelihood based IPW (ML-IPW) and augmented IPW (DR) estimators. Sample size $n=2000$.}\vskip 0.2cm 
\label{tab:cate3}
\begin{tabular}{cccccc}
\hline
$\theta_k-\theta_0$&bias&sd&MSE&$\wh{\rm sd}$&95\%\\ \hline
&\multicolumn{5}{c}{$m$, $\pi$ correct}\\ 
$k=1$&0.0147&1.2971&1.6826&1.3063&0.9490\\
$k=2$&0.0147&1.2972&1.6830&1.3059&0.9510\\
$k=3$&0.0125&1.2972&1.6830&1.3059&0.9520\\
&\multicolumn{5}{c}{$\pi$ correct}\\ 
$k=1$&0.1837&3.5871&12.9007&3.2328&0.9310\\
$k=2$&0.1936&3.3857&11.5003&3.0257&0.9220\\
$k=3$&0.1522&3.3617&11.3241&3.0269&0.9310\\
&\multicolumn{5}{c}{$m$ correct}\\ 
$k=1$&0.7568&1.4234&2.5987&1.4744&0.9450\\
$k=2$&0.7566&1.4232&2.5980&1.4740&0.9460\\
$k=3$&0.7541&1.4243&2.5975&1.4740&0.9460\\
&\multicolumn{5}{c}{$m$, $\pi$ misspecified}\\
$k=1$&0.9441&3.6714&14.3704&3.3614&0.9190\\
$k=2$&0.9392&3.4964&13.1066&3.1605&0.9140\\
$k=3$&0.8885&3.4607&12.7659&3.1639&0.9290\\
&\multicolumn{5}{c}{ML-IPW, $\pi$ correct}\\
$k=1$&-0.0998&7.1859&51.6464&7.2091&0.9460\\
$k=2$&0.1173&6.3511&40.3504&6.3572&0.9460\\
$k=3$&0.1109&6.3369&40.1689&6.3598&0.9420\\
&\multicolumn{5}{c}{DR, $m$, $\pi$ correct}\\
$k=1$&-0.006&1.208&1.459&1.229& 0.962\\
$k=2$&-0.006&1.209&1.461&1.228& 0.958\\
$k=3$&-0.008&1.208&1.460&1.228& 0.959 \\ \hline
\end{tabular}
\end{center}
\end{table}

\begin{table}[h]
    \caption{Results based on 1000 replicates for continuous treatment case, and nonlinear outcome model. Integrated absolute bias and  integrated RMSE (in parentheses). ML-IPW is the maximum likelihood based IPW estimator and CB-IPW the robust balancing-IPW method proposed (\ref{eq:betaconint}-\ref{eq:thetacon}).} 
    \label{tab:nonlinear}
    \begin{center}
    $n=500$ \\
    \vspace{0.5em}
    \begin{tabular}{c|c|cccc}
\hline
\multicolumn{2}{c|}{}&$\pi, m$ correct&$\pi$ correct&$m$ correct &none correct\\
\hline
\multicolumn{2}{c|}{IPW of Kennedy} &na  &3.33 (4.95)  &na  &3.00 (4.81)  \\
\multicolumn{2}{c|}{DR of Kennedy } &1.09 (3.31)  &2.05 (3.75)  &1.07 (3.31)  &2.55 (4.02)  \\
\hline
\multicolumn{2}{c|}{} &\multicolumn{2}{c}{$\pi$ correct}&\multicolumn{2}{c}{none correct}\\
\hline
\multirow{3}{*}{ML-IPW}
&Constant, CV    &na  &0.52 (4.52)  &na  &1.21 (4.40)  \\
&Constant, OSCV  &na  &0.39 (4.23)  &na  &1.49 (4.42)  \\
&Linear, OSCV    &na  &0.40 (4.08)  &na  &1.99 (4.45)  \\
\hline
\multirow{3}{*}{CB-IPW}
&Constant, CV    &0.38 (4.24)  &0.26 (4.32)  &1.15 (4.18)  &1.23 (4.25)  \\
&Constant, OSCV  &0.28 (4.05)  &0.31 (4.18)  &1.41 (4.26)  &1.52 (4.35)  \\
&Linear, OSCV    &0.69 (3.91)  &0.82 (4.09)  &1.86 (4.22)  &1.99 (4.34)  \\
\hline
    \end{tabular}
    
    \vspace{1em}
    $n=1000$ \\
    \vspace{0.5em}
    \begin{tabular}{c|c|cccc}
\hline
\multicolumn{2}{c|}{}&$\pi, m$ correct&$\pi$ correct&$m$ correct &none correct\\
\hline
\multicolumn{2}{c|}{IPW of Kennedy} &na  &3.15 (4.11)  &na  &2.80 (3.91)  \\
\multicolumn{2}{c|}{DR of Kennedy } &0.97 (2.60)  &1.88 (3.16)  &0.94 (2.37)  &2.36 (3.28)  \\
\hline
\multicolumn{2}{c|}{} &\multicolumn{2}{c}{$\pi$ correct}&\multicolumn{2}{c}{none correct}\\
\hline
\multirow{3}{*}{ML-IPW}
&Constant, CV    &na  &0.39 (3.26)  &na  &1.32 (3.30)  \\
&Constant, OSCV  &na  &0.46 (2.88)  &na  &1.42 (3.23)  \\
&Linear, OSCV    &na  &0.48 (2.80)  &na  &1.96 (3.41)  \\
\hline
\multirow{3}{*}{CB-IPW}
&Constant, CV    &0.27 (3.08)  &0.20 (3.15)  &1.27 (3.13)  &1.34 (3.19)  \\
&Constant, OSCV  &0.29 (2.78)  &0.20 (2.89)  &1.37 (3.08)  &1.46 (3.17)  \\
&Linear, OSCV    &0.68 (2.72)  &0.69 (2.88)  &1.85 (3.20)  &1.97 (3.32)  \\
\hline
    \end{tabular}
    
    \vspace{1em}
    $n=2000$ \\
    \vspace{0.5em}
    \begin{tabular}{c|c|cccc}
\hline
\multicolumn{2}{c|}{}&$\pi, m$ correct&$\pi$ correct&$m$ correct &none correct\\
\hline
\multicolumn{2}{c|}{IPW of Kennedy} &na  &3.02 (3.62)  &na  &2.65 (3.44)  \\
\multicolumn{2}{c|}{DR of Kennedy } &0.79 (1.83)  &1.76 (2.58)  &0.78 (1.81)  &2.37 (3.82)  \\
\hline
\multicolumn{2}{c|}{} &\multicolumn{2}{c}{$\pi$ correct}&\multicolumn{2}{c}{none correct}\\
\hline
\multirow{3}{*}{ML-IPW}
&Constant, CV    &na  &0.33 (2.44)  &na  &1.45 (2.76)  \\
&Constant, OSCV  &na  &0.56 (2.09)  &na  &1.41 (2.57)  \\
&Linear, OSCV    &na  &0.54 (1.97)  &na  &2.00 (2.89)  \\
\hline
\multirow{3}{*}{CB-IPW}
&Constant, CV    &0.22 (2.30)  &0.19 (2.43)  &1.41 (2.59)  &1.47 (2.66)  \\
&Constant, OSCV  &0.39 (1.95)  &0.26 (2.12)  &1.36 (2.39)  &1.44 (2.49)  \\
&Linear, OSCV    &0.66 (1.91)  &0.70 (2.15)  &1.91 (2.68)  &2.00 (2.81)  \\
\hline
    \end{tabular}
    
    Note: ``na" stands for ``not applicable".
    \end{center}
\end{table}

\begin{table}[htb]
    \caption{Results based on 1000 replicates for continuous treatment
      case, and linear outcome model. Integrated absolute bias and
      integrated RMSE (in parentheses).  ML-IPW is the maximum likelihood based IPW estimator and CB-IPW the robust balancing-IPW method proposed
      (\ref{eq:betaconint}-\ref{eq:thetacon}).}  
    \label{tab:linear}
    \begin{center}
    $n=500$ \\
    \vspace{0.5em}
    \begin{tabular}{c|c|cccc}
\hline
\multicolumn{2}{c|}{}&$\pi, m$ correct&$\pi$ correct&$m$ correct &none correct\\
\hline
\multicolumn{2}{c|}{IPW of Kennedy} &na  &3.02 (5.31)  &na  &2.58 (4.02)  \\
\multicolumn{2}{c|}{DR of Kennedy } &0.58 (2.60)  &0.72 (2.69)  &0.64 (2.55)  &0.90 (2.64)  \\
\hline
\multirow{3}{*}{ML-IPW}
&Constant, CV    &na  &0.26 (3.55)  &na  &0.28 (3.55)  \\
&Constant, OSCV  &na  &0.07 (3.64)  &na  &0.55 (3.74)  \\
&Linear, OSCV    &na  &0.18 (3.36)  &na  &0.68 (3.44)  \\
\hline
\multirow{3}{*}{CB-IPW}
&Constant, CV    &0.23 (3.29)  &0.17 (3.34)  &0.27 (3.21)  &0.29 (3.28)  \\
&Constant, OSCV  &0.12 (3.55)  &0.21 (3.58)  &0.53 (3.56)  &0.56 (3.58)  \\
&Linear, OSCV    &0.25 (3.23)  &0.33 (3.27)  &0.65 (3.26)  &0.68 (3.30)  \\
\hline
    \end{tabular}
    
    \vspace{0.5em}
    $n=1000$ \\
    \vspace{0.5em}
    \begin{tabular}{c|c|cccc}
\hline
\multicolumn{2}{c|}{}&$\pi, m$ correct&$\pi$ correct&$m$ correct &none correct\\
\hline
\multicolumn{2}{c|}{IPW of Kennedy} &na  &2.96 (4.82)  &na  &2.55 (3.33)  \\
\multicolumn{2}{c|}{DR of Kennedy } &0.44 (1.92)  &0.62 (1.97)  &0.48 (1.85)  &0.78 (1.98)  \\
\hline
\multirow{3}{*}{ML-IPW}
&Constant, CV    &na  &0.29 (2.55)  &na  &0.27 (2.52)  \\
&Constant, OSCV  &na  &0.10 (2.52)  &na  &0.46 (2.61)  \\
&Linear, OSCV    &na  &0.07 (2.31)  &na  &0.58 (2.43)  \\
\hline
\multirow{3}{*}{CB-IPW}
&Constant, CV    &0.23 (2.34)  &0.19 (2.39)  &0.26 (2.28)  &0.26 (2.32)  \\
&Constant, OSCV  &0.04 (2.43)  &0.05 (2.46)  &0.44 (2.46)  &0.44 (2.48)  \\
&Linear, OSCV    &0.15 (2.21)  &0.16 (2.27)  &0.56 (2.27)  &0.56 (2.31)  \\
\hline
    \end{tabular}
    
    \vspace{0.5em}
    $n=2000$ \\
    \vspace{0.5em}
    \begin{tabular}{c|c|cccc}
\hline
\multicolumn{2}{c|}{}&$\pi, m$ correct&$\pi$ correct&$m$ correct &none correct\\
\hline
\multicolumn{2}{c|}{IPW of Kennedy} &na  &2.93 (3.44)  &na  &2.43 (2.97)  \\
\multicolumn{2}{c|}{DR of Kennedy } &0.41 (1.45)  &0.60 (1.55)  &0.43 (1.40)  &0.75 (1.57)  \\
\hline
\multirow{3}{*}{ML-IPW}
&Constant, CV    &na  &0.22 (1.84)  &na  &0.32 (1.84)  \\
&Constant, OSCV  &na  &0.12 (1.79)  &na  &0.42 (1.85)  \\
&Linear, OSCV    &na  &0.09 (1.70)  &na  &0.57 (1.80)  \\
\hline
\multirow{3}{*}{CB-IPW}
&Constant, CV    &0.18 (1.72)  &0.15 (1.74)  &0.29 (1.66)  &0.29 (1.67)  \\
&Constant, OSCV  &0.08 (1.72)  &0.06 (1.76)  &0.40 (1.74)  &0.40 (1.76)  \\
&Linear, OSCV    &0.14 (1.61)  &0.14 (1.65)  &0.54 (1.68)  &0.55 (1.71)  \\
\hline
    \end{tabular}
    
    Note: ``na" stands for ``not applicable".
    \end{center}
    
\end{table}

\end{document}